\documentclass[12pt,a4j]{article}
\usepackage{amsmath,amsthm,amssymb,bm,fancybox,multirow,hhline,psfrag,comment}
\usepackage{url}
\usepackage{color}
\usepackage[dvipdfm]{graphicx}
\usepackage{natbib}
\usepackage[subrefformat=parens]{subcaption}

\def\hsymbu#1{\smash{\lower1.7ex\hbox{\huge$#1$}}}

\usepackage{indentfirst}
\usepackage{algorithm,algorithmic}
\usepackage{ulem}

\bibpunct{(}{)}{;}{a}{}{,}

\def\E{{E}}

\makeatletter
\def\eqnarray{\stepcounter {equation}\let \@currentlabel =\theequation
\global \@eqnswtrue
\global \@eqcnt \z@ \tabskip \@centering \let \\=\@eqncr
$$\halign to \displaywidth \bgroup \@eqnsel \hskip \@centering
$\displaystyle \tabskip \z@ {##}$&\global \@eqcnt \@ne \hfil
${\mbox{}##\mbox{}}$\hfil &\global \@eqcnt \tw@
$\displaystyle \tabskip \z@ {##}$\hfil \tabskip \@centering
&\llap {##}\tabskip \z@ \cr}
\makeatother

\begin{document}

\theoremstyle{plain}
\newtheorem{lemma}{Lemma}[section]
\theoremstyle{remark}
\newtheorem{remark}{Remark}[section]
\theoremstyle{example}
\newtheorem{example}{Example}[section]
\theoremstyle{lemma}
\newtheorem{prop}{Proposition}[section]

 \newcommand{\hirose}[1]{\textcolor{red}{#1}}
 \newcommand{\hiroseb}[1]{\textcolor{blue}{#1}}
 \newcommand{\hiroseu}[1]{\sout{\textcolor{red}{#1}}}

{
\begin{center}
\textbf{\Large Simple structure estimation via prenet penalization}
\end{center}
\begin{center}
\large {Kei Hirose$^{1,3}$ and Yoshikazu Terada$^{2,3}$ 
}
\end{center}

\begin{center}
{\it {\small
$^1$ Institute of Mathematics for Industry, Kyushu University,\\
744 Motooka, Nishi-ku, Fukuoka 819-0395, Japan \\

\vspace{1.2mm}

$^2$ Division of Mathematical Science for Social Systems, Graduate School of Engineering Science, Osaka University,\\
1-3, Machikaneyama-cho, Toyonaka, Osaka 560-8531, Japan \\

\vspace{1.2mm}

$^3$ RIKEN Center for Advanced Intelligence Project, 1-4-1 Nihonbashi, Chuo-ku, Tokyo 103-0027, Japan \\
}}
{\it {\small E-mail: hirose@imi.kyushu-u.ac.jp, terada@sigmath.es.osaka-u.ac.jp
}}
\end{center}

\vspace{1.5mm}

\begin{abstract}
We propose a {\it prenet} ({\it pr}oduct {\it e}lastic {\it net}), which is a new penalization method for factor analysis models.  The penalty is based on the product of a pair of elements in each row of the loading matrix.  The prenet not only shrinks some of the factor loadings toward exactly zero, but also enhances the simplicity of the loading matrix, which plays an important role in the interpretation of the common factors.  In particular, with a large amount of prenet penalization, the estimated loading matrix possesses a perfect simple structure, which is known as a desirable structure in terms of the simplicity of the loading matrix.  Furthermore, the perfect simple structure estimation via the prenet turns out to be a generalization of the $k$-means clustering of variables.  On the other hand, a mild amount of the penalization approximates a loading matrix estimated by the quartimin rotation, one of the most commonly used oblique rotation techniques.  Thus, the proposed penalty bridges a gap between the perfect simple structure and the quartimin rotation.  Monte Carlo simulation is conducted to investigate the performance of the proposed procedure.  Three real data analyses are given to illustrate the usefulness of our penalty.
 \end{abstract}

 \section{Introduction}
 Factor analysis investigates the correlation structure of high-dimensional observed variables by construction of a small number of latent variables called common factors.  Factor analysis can be considered as a soft clustering of variables, in which each factor corresponds to a cluster and observed variables are categorized into overlapping clusters.  For interpretation purposes,  it is desirable for the observed variables to be well-clustered \citep{yamamoto2013cluster}.  In particular, the perfect simple structure (e.g., \citealp{bernaards2003orthomax,jennrich2004rotation}), wherein each row of the loading matrix has at most one nonzero element, provides a non-overlapping clustering of variables in the sense that variables that correspond to nonzero elements of the $j$th column of the loading matrix belong to the $j$th cluster.

Conventionally, a well-clustered structure of the loading matrix is found by rotation techniques, such as the varimax rotation \citep{kaiser1958varimax} and the promax rotation \citep{hendrickson1964promax}.  The problem with the rotation technique is that it cannot produce a sufficiently sparse solution in some cases \citep{hirose2015sparse}, because the loading matrix must be found among a set of unpenalized maximum likelihood estimates.  To obtain sparser solutions than the factor rotation, we employ a penalization method. It is shown that the penalization is a generalization of the rotation techniques, and can produce sparser solutions than the rotation methods \citep{hirose2015sparse}.   Typically, many researchers use the $L_1$-type penalization, such as the lasso \citep{Tibshirani:1996}, the adaptive lasso \citep{Zou:2006}, and the minimax concave penalty (e.g., \citealp{Zhang:2010}).   Examples include \citet{Choietal:2011,Ningetal:2011,srivastava2014expandable,hirose2015sparse,Trendafilov2017}. The $L_1$ penalization shrinks some of the factor loadings toward exactly zero, which might produce a more interpretable loading matrix.

However, the $L_1$ penalization procedures introduce two fundamental issues.  First, the lasso-type sparse estimation is not guaranteed to produce a well-clustered structure of the loading matrix simply because it is sparse.  For example, with the lasso, a great amount of penalization leads to a zero matrix, which implies there are no cluster structures.  Even when an appropriate value of the tuning parameter is given, the estimated loading matrix is not guaranteed to possess the well-clustered structure, such as perfect simple structure.  The second issue is that the $L_1$ penalization cannot often approximate a true loading matrix when it is not sufficiently sparse; with the lasso, some of the factor loadings whose true values are close---but not very close---to zero are estimated as zero values, and this misspecification can often cause a significant negative effect on the estimation of other factor loadings \citep{hirose2014estimation}.  

To handle the above issues, we propose a {\it prenet} ({\it pr}oduct {\it e}lastic {\it net}) penalty, which is based on the product of a pair of parameters in each row of the loading matrix.  A remarkable feature of the prenet is that a large amount of penalization leads to the perfect simple structure.  The existing $L_1$-type penalization methods do not have that significant property.  Furthermore, the perfect simple structure estimation via the prenet penalty is shown to be a generalization of the $k$-means variables clustering.  On the other hand, with a mild amount of prenet penalization, the estimated loading matrix is approximated by that obtained using the quartimin rotation, a widely used oblique rotation method.  The quartimin criterion can often estimate a non-sparse loading matrix appropriately, so that the second problem of the lasso-type penalization mentioned above is addressed.  We employ the generalized expectation and maximization (GEM) algorithm and the coordinate descent algorithm (e.g., \citealp{Friedmanetal:2010}) to obtain the prenet estimator.  The proposed algorithm monotonically decreases the objective function at each iteration.  The performance of the prenet penalization is investigated through the Monte Carlo simulation.  We apply the proposed method to three datasets: personality data (big 5 data), handwritten digits data, and resting-state fMRI data.  The proposed procedure is available for use in the {\tt R} package {\tt fanc}, which is available at \url{http://cran.r-project.org/web/packages/fanc}.

The remainder of this paper is organized as follows. Section 2 describes the estimation of the factor analysis model via penalization.  In Section 3, we introduce the prenet penalty and provide an illustrative example.  Section 4 describes several properties of the prenet penalty, including its relationship with the quartimin criterion. Section 5 presents an estimation algorithm, which is based on the GEM and coordinate descent algorithms, to obtain the prenet solutions.  In Section 6, we conduct a Monte Carlo simulation to investigate the performance of the prenet penalization.  Section 7 illustrates the usefulness of our proposed procedure through three real data analyses.  Section 8 discusses the results and concludes.

\section{Estimation of the factor model via the penalization method}
Let $\bm{X} = (X_1,\dots,X_p)^T$ be a $p$-dimensional observed random vector with mean vector $\bm{0}$ and variance--covariance matrix $\bm{\Sigma}$.  The factor analysis model is
\begin{equation*}
\bm{X} = \bm{\Lambda} \bm{F}+\bm{\varepsilon}, \label{model1}
\end{equation*}
where $\bm{\Lambda} = (\lambda_{ij})$ is a $p \times m$  loading matrix, $\bm{F} = (F_1,\cdots,F_m)^T$ is a random vector of common factors, and $\bm{\varepsilon}  = (\varepsilon_1,\cdots, \varepsilon_p)^T$ is a random vector of unique factors.  It is assumed that $\E(\bm{F} ) = \bm{0}$, $\E(\bm{\varepsilon} ) = \mathbf{0}$, $\E(\bm{F}\bm{F}^T) = \bm{I}_m$, $\E(\bm{\varepsilon} \bm{\varepsilon} ^T) = \bm{\Psi}$, and $\E(\bm{F} \bm{\varepsilon} ^T) = \bm{O}$, where $\bm{I}_m$ is an identity matrix of order $m$, and $\bm{\Psi}$ is a $p \times p$ diagonal matrix whose diagonal elements are referred to as unique variances, $\psi_{i}$.  Under these assumptions, the variance--covariance matrix of observed random vector $\bm{X}$ is given by $\bm{\Sigma} = \bm{\Lambda} \bm{\Lambda}^T+\bm{\Psi}$.

Let $\bm{x}_1,\cdots,\bm{x}_n$ be $n$ observations and $\bm{S} = (s_{ij})$ be the corresponding sample covariance matrix.  We estimate the model parameter by minimizing the penalized loss function $\ell_{\rho}(\bm{\Lambda},\bm{\Psi})$ given by
\begin{equation}
	\ell_{\rho}(\bm{\Lambda},\bm{\Psi}) = \ell(\bm{\Lambda},\bm{\Psi}) + \rho P(\bm{\Lambda}), \label{eq:pf}
\end{equation}
where $\ell(\bm{\Lambda},\bm{\Psi})$ is a loss function, $P(\bm{\Lambda})$ is a penalty function, and $\rho > 0$ is a tuning parameter.  Two popular loss functions are given as follows.
\begin{description}
	\item[Quadratic loss:] A general form of the quadratic loss is given by
\begin{equation*}
	\ell_{\rm QL}(\bm{\Lambda},\bm{\Psi}) = \|\bm{\Gamma}^{-1} (\bm{S} - \bm{\Lambda}\bm{\Lambda}^T - \bm{\Psi}) \|^2, \label{GLS}
\end{equation*}
where $\bm{\Gamma}$ is an arbitrary matrix.  When $\bm{\Gamma} = \bm{I}$, $\ell_{\rm QL}(\bm{\Lambda},\bm{\Psi})$ becomes a square loss function. $\bm{\Gamma} = \bm{S}^{-1}$ results in the generalized square loss function.
	\item[Discrepancy function:] Another popular loss function is the discrepancy function
\begin{eqnarray}
\ell_{\rm ML}(\bm{\Lambda},\bm{\Psi}) = \frac{1}{2} \left\{\mathrm{tr}(\bm{\Sigma}^{-1} \bm{S}) - \log |\bm{\Sigma}^{-1}\bm{S}| - p \right\}. \label{taisuuyuudo}
\end{eqnarray}
Assume that the observations $\bm{x}_1,\cdots,\bm{x}_n$ are drawn from the $p$-dimensional normal population $N_p(\bm{\mu},\bm{\Sigma})$ with $\bm{\Sigma} = \bm{\Lambda} \bm{\Lambda}^T+\bm{\Psi}$.  The minimizer of $\ell_{\rm ML}(\bm{\Lambda},\bm{\Psi})$ is the maximum likelihood estimate.  Note that $\ell_{\rm ML}(\bm{\Lambda},\bm{\Psi})\le 0$ for any $\bm{\Lambda}$ and $\bm{\Psi}$, and $\ell_{\rm ML}(\bm{\Lambda},\bm{\Psi}) = 0$ if and only if $\bm{\Lambda}\bm{\Lambda}^T + \bm{\Psi} = \bm{S}$.
\end{description}
Hereafter, we use a discrepancy function as a loss function, unless otherwise noted. It is worth noting that our proposed penalty, described in Section \ref{sec:prenet}, can be directly applied to many other loss functions.

The factor analysis model has a rotational indeterminacy; both $\bm{\Lambda}$ and $\bm{\Lambda} \mathbf{T}$ generate the same covariance matrix $\bm{\Sigma}$, where $\bm{T}$ is an arbitrary orthogonal matrix.  Thus, when $\rho = 0$, the solution that minimizes (\ref{eq:pf}) is not uniquely determined.  However, when $\rho > 0$, the solution may be uniquely determined when an appropriate penalty $P(\bm{\Lambda})$ is chosen.  An example is the lasso penalty \citep{Tibshirani:1996}, given by
$
		P(\bm{\Lambda}) = \sum_{i = 1}^p\sum_{j = 1}^m|\lambda_{ij}|. \label{lasso penalty}
$
The lasso-type penalization produces a sparse solution, that is, some of the estimates of factor loadings become exactly zero.  

The penalty $P(\bm{\Lambda})$ is referred to as separable if it is written as $P(\bm{\Lambda}) = \sum_{i = 1}^p\sum_{j = 1}^m P(|\lambda_{ij}|)$.  Many existing penalties, including the lasso, elastic net, and SCAD penalties, are separable.  The most popular nonseparable penalty would be the fused lasso \citep{tibshirani2005sparsity}, in which the penalty is based on the {\it difference} of the coefficients.
\begin{remark}
	There are several latent variable models related to the standard factor model.  Here, we describe three models.
	\begin{enumerate}
		\item We can assume a factor correlation (i.e., $E[\bm{F}\bm{F}^T] = \bm{\Phi}$) and estimate it by the penalized maximum likelihood method \citep{hirose2014estimation}.
		\item The approximate factor model (e.g., \citealp{stock2002forecasting}), in which $\bm{\Psi}$ does not have to be a diagonal matrix, may be more flexible than the standard factor model.
		\item $\bm{\Psi} = \sigma^2\bm{I}$ corresponds to the probabilistic principal component analysis \citep{tipping1999probabilistic}.  This fact implies that the factor analysis is viewed as a generalization of principal component analysis.
	\end{enumerate}
Our proposed penalty, presented in Section \ref{sec:prenet}, can be directly applied to a wide variety of latent variable models, including the above three models.
\end{remark}

\section{Prenet penalty}\label{sec:prenet}
We propose the {\it prenet} ({\it pr}oduct {\it e}lastic {\it net}) penalty
\begin{equation}
P(\bm{\Lambda}) = \sum_{i = 1}^p \sum_{j = 1}^{m-1} \sum_{k > j}  \left\{  \gamma  |\lambda_{ij}||\lambda_{ik}| +  \frac{1}{2} (1-\gamma)  \lambda_{ij}^2\lambda_{ik}^2 \right\}, \label{prenet penalty}
\end{equation}
where $\gamma \in [0,1]$ is a tuning parameter.  The most significant feature of the prenet penalty is that it is based on the {\it product} of a pair of parameters.  It is shown that the prenet penalty is not separable.

When $\gamma = 0$, the prenet penalty is equivalent to the quartimin criterion \citep{carroll1953analytical}, a widely used oblique rotation criterion in factor rotation.  As is the case with the quartimin rotation, the prenet penalty in (\ref{prenet penalty}) eliminates the rotational indeterminacy and contributes significantly to the estimation of the simplicity of the loading matrix.  When $\gamma > 0$, the prenet penalty includes products of absolute values of factor loadings, producing factor loadings that are {\it exactly} zero.  Therefore, with an appropriate value of $\gamma$, the prenet penalty enhances both the simplicity and the sparsity of the loading matrix.

\subsection{Comparison with the elastic net penalty}
The prenet penalty is similar to the elastic net penalty \citep{ZouandHastie:2005}
\begin{equation}
P(\bm{\Lambda}) = \sum_{i = 1}^p \sum_{j = 1}^{m} \left\{  \gamma  |\lambda_{ij}| + \frac{1}{2} (1-\gamma)  \lambda_{ij}^2  \right\},\label{eq:enet}
\end{equation}
which is a hybrid of the lasso penalty (first term) and the ridge penalty (second term).  Although the elastic net penalty is similar to the prenet penalty, there is a fundamental difference between these two penalties; the elastic net is a separable penalty, but the prenet is based on the product of a pair of parameters.

Figure \ref{fig:penalties} shows the penalty functions of the prenet ($P(x,y) = \gamma |x||y| + (1-\gamma)x^2y^2/2$) and the elastic net ($P(x,y) = \gamma (|x|+|y|) + (1-\gamma)(x^2+y^2)/2$) when $\gamma = 0.7$.  Clearly, the prenet penalty is a nonconvex function.  A significant difference between the prenet and the elastic net is that although the prenet penalty becomes zero when {\it either} $x$ or $y$ attains zero, the elastic net penalty becomes zero only when {\it both} $x = 0$ and $y = 0$.  Therefore, for a two-factor model, either $\lambda_{i1}$ or $\lambda_{i2}$ tends to be close to zero with the prenet penalty, which leads to a perfect simple structure.  On the other hand, the elastic net tends to produce estimates in which both $\lambda_{i1}$ and $\lambda_{i2}$ are small.
 \begin{figure}[!t]
\centering
\includegraphics[width=6.0cm]{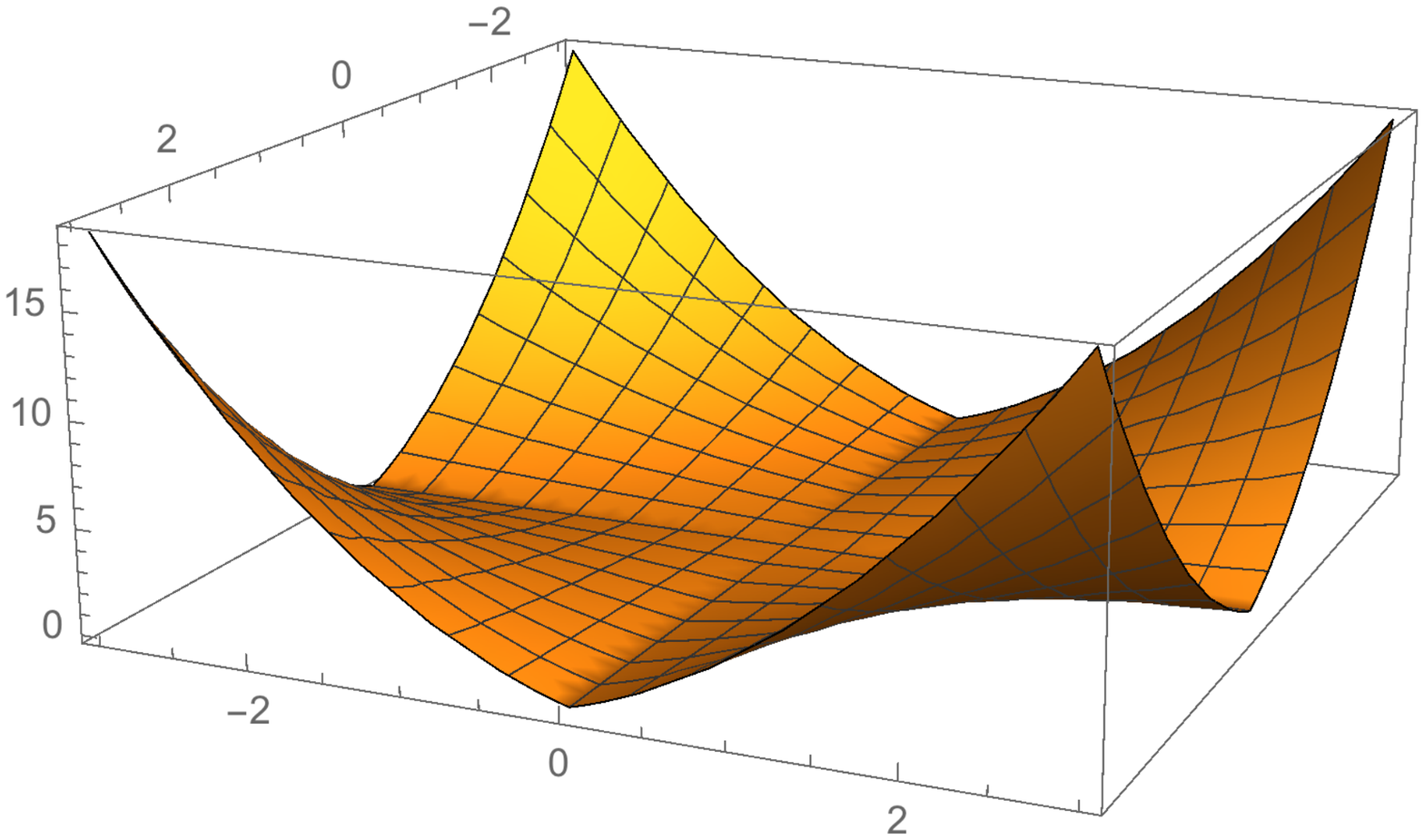}\hspace{10mm}
  \includegraphics[width=6.0cm]{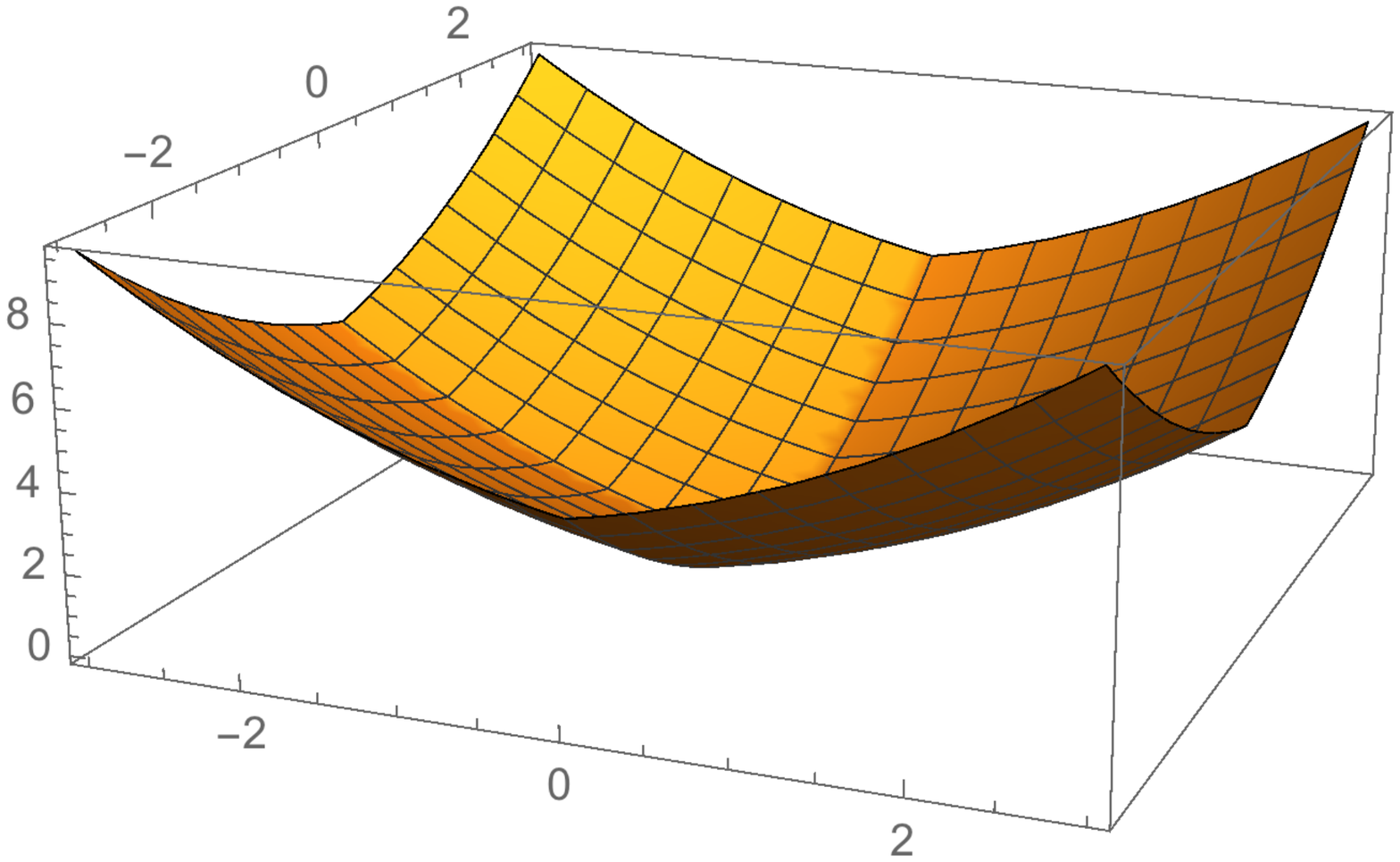}
  \caption{Penalty functions of the prenet (left-hand side) and the elastic net (right-hand side) with $\gamma = 0.7$.}	\label{fig:penalties}
\end{figure}

With the prenet penalty, the second term of (\ref{prenet penalty}) allows the estimation of the simplicity of the loading matrix.  However, the second term of the elastic net penalty in (\ref{eq:enet}) (i.e., ridge penalty) does not contribute in any way to the estimation of the simplicity of the loading matrix.  In fact, the ridge penalty can be expressed as
 \begin{equation*}
\sum_{i = 1}^p \sum_{j = 1}^{m} \lambda_{ij}^2 = {\rm tr}(\bm{\Lambda}^T\bm{\Lambda}) = {\rm tr}(\bm{\Lambda}^T \bm{T} \bm{T}^T\bm{\Lambda})
\end{equation*}
for any orthogonal matrix $\bm{T}$, which implies the rotational indeterminacy cannot be eliminated with the ridge penalty.  On the other hand, the lasso makes some of the coefficients move toward exactly zero, which leads to an interpretable loading matrix.  Nevertheless,  the sparse estimation with the lasso cannot often estimate a well-cluster structure. For example, when the true loading matrix is not sufficiently sparse, the lasso often estimates a loading matrix that is completely different from the true one \citep{hirose2014estimation}.   We provide a simple numerical example in the next Subsection to illustrate this point.

\subsection{Illustrative example}\label{sec:illustration}
Assume that the true loading matrix is
\begin{equation}
	\bm{\Lambda}_d=
	\begin{pmatrix}
		0.9 & 0.8 & 0.7 & 0.2 & 0.2 & 0.2\\
		0.2 & 0.2 & 0.2 & 0.9 & 0.8 & 0.7
	\end{pmatrix}
	^T.\label{noisy}
\end{equation}
Here, ``$d$" in $\bm{\Lambda}_d$ denotes density, because the loading matrix does not include zero values.  We construct a covariance matrix $\bm{\Sigma} = \bm{\Lambda}_d\bm{\Lambda}_d^T + \bm{\Psi}$ with $\bm{\Psi} = {\rm diag} (\bm{I} - \bm{\Lambda}_d\bm{\Lambda}_d^T)$, and then generate 50 samples from $N(\bm{0},\bm{\Sigma})$.  In many simulation studies of the factor model (e.g., \citealp{lopes2004bayesian}), some of the true factor loadings are exactly zero, as follows:
\begin{equation}
	{\bm{\Lambda}}_s =
	\begin{pmatrix}
		0.9 & 0.8 & 0.7& 0.0 & 0.0 & 0.0\\
		0.0 & 0.0 & 0.0 & 0.9 & 0.8 & 0.7
	\end{pmatrix}
	^T. \label{desirable}
\end{equation}
Here, ``$s$" in $\bm{\Lambda}_s$ denotes sparsity.  In this numerical example, we use $\bm{\Lambda}_d$ instead of $\bm{\Lambda}_s$.  This is because in many applications, some of the factor loadings can be nearly---but not exactly---zero.

With the penalization procedure, we expect that
\begin{description}
  \item[(i)]for large $\rho$, the estimated loading matrix is close to (\ref{desirable}),
  \item[(ii)]for small $\rho$, we obtain an estimate close to (\ref{noisy}).
\end{description}

Table \ref{table:illustration} shows the loading matrices estimated by the elastic net for various values of $\rho$.
\begin{table}
\caption{\label{table:illustration}Loading matrices estimated by the lasso for various values of $\rho$.}
\centering
\begin{tabular}{r|rr|rr|rr|rr|rr|rr}
  \hline
\multicolumn{1}{r}{}  &\multicolumn{6}{c}{$\gamma = 1$}&\multicolumn{6}{c}{$\gamma = 0.01$}\\
\multicolumn{1}{r}{} &\multicolumn{2}{c}{$\rho = 0.28$}&\multicolumn{2}{c}{$\rho = 0.1$}&\multicolumn{2}{c}{$\rho = 0.01$}&\multicolumn{2}{c}{$\rho = 1.0$}&\multicolumn{2}{c}{$\rho = 0.1$}&\multicolumn{2}{c}{$\rho = 0.01$}\\
\multicolumn{1}{r}{} & F1 & \multicolumn{1}{r}{F2}  & F1 & \multicolumn{1}{r}{F2} & F1 & \multicolumn{1}{r}{F2} & F1 & \multicolumn{1}{r}{F2} & F1 & \multicolumn{1}{r}{F2}& F1 & \multicolumn{1}{r}{F2} \\
  \hline
V1  & 0.63 & 0.00 & 0.74 & 0.00 & 0.85 & 0.00 & 0.52 & $-$0.00 & 0.77 & $-$0.00 & 0.86 & 0.01 \\
  V2 & 0.66 & 0.00 & 0.76 & 0.00 & 0.86 & 0.00 & 0.53 & 0.00 & 0.79 & 0.01 & 0.87 & 0.02 \\
  V3 & 0.46 & 0.00 & 0.59 & 0.04 & 0.70 & 0.08 & 0.39 & 0.06 & 0.62 & 0.08 & 0.70 & 0.10 \\
  V4 & 0.20 & 0.52 & 0.35 & 0.58 & 0.50 & 0.64 & 0.28 & 0.40 & 0.43 & 0.60 & 0.50 & 0.66 \\
  V5 & 0.09 & 0.60 & 0.24 & 0.68 & 0.38 & 0.74 & 0.21 & 0.43 & 0.32 & 0.69 & 0.38 & 0.75 \\
  V6 & 0.10 & 0.46 & 0.26 & 0.55 & 0.40 & 0.62 & 0.22 & 0.36 & 0.34 & 0.57 & 0.40 & 0.63 \\
   \hline
\end{tabular}	
\end{table}
\begin{table}
\caption{\label{table:illustration2}Loading matrices estimated by the prenet for various values of $\rho$.}
\centering
\begin{tabular}{r|rr|rr|rr|rr|rr|rr}
  \hline
\multicolumn{1}{r}{}  &\multicolumn{6}{c}{$\gamma = 1$}&\multicolumn{6}{c}{$\gamma = 0.01$}\\
\multicolumn{1}{r}{} &\multicolumn{2}{c}{$\rho = 0.4$}&\multicolumn{2}{c}{$\rho = 0.2$}&\multicolumn{2}{c}{$\rho = 0.01$}&\multicolumn{2}{c}{$\rho = 43$}&\multicolumn{2}{c}{$\rho = 0.5$}&\multicolumn{2}{c}{$\rho = 0.02$}\\
\multicolumn{1}{r}{} & F1 & \multicolumn{1}{r}{F2}  & F1 & \multicolumn{1}{r}{F2} & F1 & \multicolumn{1}{r}{F2} & F1 & \multicolumn{1}{r}{F2} & F1 & \multicolumn{1}{r}{F2}& F1 & \multicolumn{1}{r}{F2} \\
  \hline
V1 & 0.88 & 0.00 & 0.83 & 0.00 & 0.86 & 0.00 & 0.88 & 0.00 & 0.81 & 0.15 & 0.84 & 0.21 \\
  V2 & 0.87 & 0.00 & 0.85 & 0.00 & 0.88 & 0.00 & 0.87 & 0.00 & 0.82 & 0.16 & 0.85 & 0.22 \\
  V3 & 0.71 & 0.00 & 0.68 & 0.04 & 0.71 & 0.08 & 0.71 & 0.00 & 0.64 & 0.20 & 0.67 & 0.26 \\
  V4 & 0.00 & 0.83 & 0.32 & 0.64 & 0.51 & 0.65 & 0.00 & 0.83 & 0.26 & 0.72 & 0.34 & 0.76 \\
  V5 & 0.00 & 0.85 & 0.19 & 0.75 & 0.39 & 0.75 & 0.00 & 0.85 & 0.14 & 0.80 & 0.20 & 0.83 \\
  V6 & 0.00 & 0.76 & 0.22 & 0.62 & 0.40 & 0.63 & 0.00 & 0.76 & 0.18 & 0.68 & 0.25 & 0.71 \\
   \hline
\end{tabular}
\end{table}
With the lasso penalty (i.e., $\gamma = 1$), when $\rho > 0.28$, we obtain a one-factor model: the largest value that provides a two-factor model with the lasso is $\rho = 0.28$.  In this case, $\hat{\lambda}_{41}$, $\hat{\lambda}_{51}$, and $\hat{\lambda}_{61}$ are nonzero, which means (i) is not satisfied.  When $\rho$ is small, $\hat{\lambda}_{12}$, $\hat{\lambda}_{22}$, and $\hat{\lambda}_{32}$ are still close to zero, but $\hat{\lambda}_{41}$, $\hat{\lambda}_{51}$, and $\hat{\lambda}_{61}$ become much larger than the true values.  Estimating some coefficients toward {\it exactly} zero makes other small coefficients larger than expected.  As a result, (ii) is not satisfied with the lasso.  When $\gamma = 0.01$, we obtain similar results, and thus, the ridge penalty does not make any contribution to the approximation of the true loading matrix.

The loading matrices estimated by the prenet penalty are given in Table \ref{table:illustration2}.  $\gamma = 1$ implies the second term in (\ref{prenet penalty}), $\sum_{i,j,k} \lambda_{ij}^2 \lambda_{ik}^2$, is not included.  When $\gamma = 1$, the prenet is able to produce a solution that is very close to (\ref{desirable}) for large $\rho$.  When $\rho$ is small, however, we obtain a tendency similar to the lasso; $\hat{\lambda}_{41}$, $\hat{\lambda}_{51}$, and $\hat{\lambda}_{61}$ are larger than the true values.  Therefore, (i) is satisfied but (ii) is not when $\gamma = 1$.

When $\gamma = 0.01$, the second term in (\ref{prenet penalty}), $\sum_{i,j,k} \lambda_{ij}^2 \lambda_{ik}^2$, is included in the prenet penalty.  When $\rho$ is large, we obtain a loading matrix that is similar to (\ref{desirable}).  Furthermore, as $\rho$ reduces, we obtain a loading matrix that is close to the true loading matrix in (\ref{noisy}).  Thus, the prenet penalty with $\gamma = 0.01$ satisfies both (i) and (ii).

\section{Properties of the prenet penalty}
\subsection{Perfect simple structure} \label{sec:pss}
Most existing penalties, such as the lasso, shrink all coefficients toward zero when the tuning parameter $\rho$ is sufficiently large; we usually obtain $\hat{\bm{\Lambda}} = \bm{0}$ when $\rho \rightarrow \infty$.  However, the following proposition shows that the prenet penalty does not shrink some of the elements toward zero even when $\rho$ is sufficiently large. 
\begin{prop}\label{prop:pss}
	Assume that we use the prenet penalty with $\gamma \in (0,1]$.  As $\rho \rightarrow \infty$, the estimated loading matrix possesses the perfect simple structure, that is, each row has at most one nonzero element.
\end{prop}
\begin{proof}
	As $\rho \rightarrow \infty$, $P(\hat{\bm{\Lambda}})$ must satisfy $P(\hat{\bm{\Lambda}}) \rightarrow 0$.  Otherwise, the second term of (\ref{eq:pf}) diverges.  $P(\hat{\bm{\Lambda}}) = 0$ implies $ \hat{\lambda}_{ij}\hat{\lambda}_{ik} = 0$ for any $j \neq k$.  Therefore, the $i$th row of $\bm{\Lambda}$ has at most one nonzero element.
	\end{proof}
The perfect simple structure is known as a desirable property in the literature on factor analysis, because it is very easy to interpret the estimated loading matrix (e.g., \citealp{bernaards2003orthomax}).  When $\rho$ reduces, the estimated loading matrix can be far from the perfect simple structure but the goodness of fit to the model is improved.  
\subsubsection{Relationship with $k$-means variables clustering}\label{sec:kmeans}
The perfect simple structure corresponds to variables clustering, that is, variables that correspond to nonzero elements of the $j$th column of the loading matrix belong to the $j$th cluster.  One of the most popular cluster analyses is the $k$-means.  In this Subsection, we investigate the relationship between the prenet solution with $\rho \rightarrow \infty$ and the $k$-means variables clustering.

Let $\bm{X}_n$ be an $n \times p$ data matrix.  $\bm{X}_n$ can be expressed as $\bm{X}_n = (\bm{x}_1^*,\dots,\bm{x}_p^*)$, where $\bm{x}_i^*$ is the $i$th column vector of $\bm{X}_n$.  We consider the problem of the variables clustering of $\bm{x}_1^*,\dots,\bm{x}_p^*$ by the $k$-means.  Let $C_j$ $(j = 1,\dots,m)$ be a subset of indices of variables that belong to the $j$th cluster.  The objective function of the $k$-means is
\begin{equation}
	\sum_{j = 1}^m \sum_{i \in C_j}\|\bm{x}_i^* - \bm{\mu}_j\|^2 =  \sum_{i = 1}^ps_{ii} - \sum_{j=1}^m \frac{1}{p_j}\sum_{i \in C_j}\sum_{i' \in C_j}s_{ii'}, \label{f_kmeans}
\end{equation}
where $p_j = \#\{C_j\}$, $\bm{\mu}_j =  \frac{1}{p_j}\sum_{i \in C_j} \bm{x}_i^*$, and recall that $s_{ii'}$ is expressed as $s_{ii'} = \bm{x}_i^{*T}\bm{x}_{i'}^*$.  Let $\bm{\Lambda} = (\lambda_{ij})$ be a $p \times m$ indicator variables matrix given by
\begin{equation}
	\lambda_{ij} =
	\left\{
	\begin{array}{rr}
		1/\sqrt{p_j} & i \in C_j, \\
		0 & i \notin C_j.
	\end{array}
	\right. \label{q_ij}
\end{equation}
Using the fact that $\bm{\Lambda}^T\bm{\Lambda} = \bm{I}_m$, the $k$-means variables clustering using (\ref{f_kmeans}) is equivalent to \citep{ding2005equivalence}.
\begin{equation}
	\min_{\bm{\Lambda}} \|\bm{S} - \bm{\Lambda}\bm{\Lambda}^T\|^2, \ \mbox{ subject to } (\ref{q_ij}). \label{kmeans_equivalent}
\end{equation}
We consider slightly modifying the condition on $\bm{\Lambda}$ in (\ref{q_ij}) to
\begin{equation}
	\lambda_{ij} \lambda_{ik} = 0 \ (j \ne k)\mbox{ and } \bm{\Lambda}^T\bm{\Lambda} = \bm{I}_m. \label{q_ij_adj}
\end{equation}
 The modified $k$-means problem is then given as
\begin{equation}
	\min_{\bm{\Lambda}} \|\bm{S} - \bm{\Lambda}\bm{\Lambda}^T\|^2 \mbox{ subject to }  (\ref{q_ij_adj}).\label{problem_kmeansgen}
\end{equation}
Note that condition (\ref{q_ij_adj}) is milder than (\ref{q_ij}): if $\bm{\Lambda}$ satisfies (\ref{q_ij}), we obtain (\ref{q_ij_adj}).  The reverse does not hold; with (\ref{q_ij_adj}), the nonzero elements for each column do not have to be equal.  Therefore, the modified $k$-means in (\ref{problem_kmeansgen}) may capture a more complex structure than the original $k$-means.

\begin{prop}\label{prop_kmeans}
Assume that $\bm{\Psi} = \alpha \bm{I}$ and $\alpha$ is given.  Suppose that $\bm{\Lambda}$ satisfies $\bm{\Lambda}^T\bm{\Lambda} = \bm{I}_m$.  The prenet solution with $\rho \rightarrow \infty$ is then obtained by (\ref{problem_kmeansgen}).
\end{prop}
\begin{proof}
	The proof appears in \ref{prop_kmeans:app}.
\end{proof}
The above proposition shows that the prenet solution with $\rho \rightarrow \infty$ is a generalization of the problem (\ref{problem_kmeansgen}).  As mentioned above, the problem (\ref{problem_kmeansgen}) is a generalization of the $k$-means problem in (\ref{kmeans_equivalent}).  Therefore, the perfect simple structure estimation via the prenet is a generalization of the $k$-means variables clustering.

\subsection{Relationship with quartimin rotation}\label{sec:qmin}
As described in Section \ref{sec:prenet}, the prenet penalty is a generalization of the quartimin criterion \citep{carroll1953analytical}; setting $\gamma = 0$ to the prenet penalty in (\ref{prenet penalty}) leads to the quartimin criterion
\begin{equation*}
	P_{\rm qmin}(\bm{\Lambda}) = \sum_{i = 1}^p \sum_{j = 1}^{m-1 } \sum_{k > j}    \lambda_{ij}^2\lambda_{ik}^2.
\end{equation*}

The quartimin criterion is usually used in the factor rotation.  The solution of quartimin rotation method, say $\hat{\bm{\theta}}_q=(\hat{\bm{\Lambda}}_q,\hat{\bm{\Psi}}_q)$, is obtained by two-step procedure.  First, we calculate an unpenalized estimator, say $\hat{\bm{\theta}}=(\hat{\bm{\Lambda}},\hat{\bm{\Psi}})$. $\hat{\bm{\theta}}$ satisfies $\displaystyle {\ell}(\hat{\bm{\theta}})=\min_{\bm{\theta}} {\ell}(\bm{\theta})$.  Note that  $\hat{\bm{\theta}}$ is not unique because of the rotational indeterminacy.  The second step is the minimization of the quartimin criterion with a restricted parameter space given by $\{ \bm{\theta}|\ell(\bm{\theta}) = \min_{\bm{\theta}}\ell(\bm{\theta}) \}$.  \citet{hirose2015sparse} showed that the solution of the quartimin rotation, $\hat{\bm{\theta}}_q$, can be obtained by 
 \begin{equation}
\min_{\bm{\theta}}P_{\rm qmin}(\bm{\Lambda}), \mbox{ subject to} \quad   \ell(\bm{\theta}) = \ell(\hat{\bm{\theta}})\label{problem_rotation_mle2}
\end{equation}
under the condition that the unpenalized estimate of loading matrix $\hat{\bm{\Lambda} }$ is unique if the indeterminacy of the rotation in $\hat{\bm{\Lambda} }$ is excluded.  Note that it is not easy to check this condition, but several necessary conditions of the identifiability are provided (e.g., Theorem 5.1 in \citealp{anderson1956statistical}.)

Now, we show a basic asymptotic result of the prenet solution, 
from which we can see that the prenet solution is a generalization of the quartimin rotation.  
Let $(\Theta,d)$ be a compact parameter space and $(\Omega,\mathcal{F},\mathbb{P})$ be a probability space.  
Suppose that for any $(\bm{\Lambda},\bm{\Psi})\in\Theta$ and any $\bm{T}\in \mathcal{O}(m)$, we have $(\bm{\Lambda} \bm{T} ,\bm{\Psi}) \in \Theta$, where $\mathcal{O}(m)$ is a set of  $m\times m$ orthonormal matrices. 
Let $\bm{X}_1,\dots,\bm{X}_n$ denote independent $\mathbb{R}^p$-valued random variables with the common population distribution $\mathbb{P}$.
Now, it is required that we can rewrite the empirical loss function and the true loss function 
as $\ell(\bm{\theta}):=\sum_{i=1}^n q(\bm{X}_i;\bm{\theta})/n$ and 
$\ell_\ast(\bm{\theta}):=\int q(\bm{x};\bm{\theta})\,\mathbb{P}(d\bm{x})$, respectively.  Note that the function $q(\bm{x};\bm{\theta})$ can be a logarithm of density function of normal distribution when  $\ell(\bm{\theta})$ is the discrepancy function, but any other functions that satisfy regularity conditions described in Proposition \ref{prop_qmin} can be used. 
Let $\hat{\bm{\theta}}_\rho=(\hat{\bm{\Lambda}}_\rho,\hat{\bm{\Psi}}_\rho)$ denote 
an arbitrary measurable prenet estimator which satisfies 
$\ell(\hat{\bm{\theta}}_\rho)
+\rho P(\hat{\bm{\Lambda}}_\rho)=\min_{\bm{\theta}\in \Theta} \ell(\bm{\theta})+\rho P(\bm{\Lambda})$.
The following proposition shows that the prenet estimator converges almost surely to 
a {\it true} parameter which minimizes the quartimin criterion when $\rho\rightarrow 0$ as $n\rightarrow \infty$.
\begin{prop} \label{prop_qmin}
Assume the following conditions: 
\begin{itemize}
\item For each $\bm{x}\in \mathbb{R}^p$, function $q(\bm{x};\bm{\theta})$ on $\Theta$ is continuous. 
\item There exists a $\mathbb{P}$-integrable function $g(\bm{x})$ such that for all $\bm{x}\in \mathbb{R}^p$ and for all $\bm{\theta}\in \Theta$
$|q(\bm{x};\bm{\theta})|\le g(\bm{x})$.
\end{itemize}
We denote by $\Theta_q^\ast$ a set of true solutions of the following quartimin problem:
$$
\min_{(\bm{\Lambda},\bm{\Psi})\in \Theta}P_\mathrm{qmin}(\bm{\Lambda})\;\text{ subject to }\;
\ell_\ast(\bm{\theta})=\min_{\bm{\theta}\in \Theta}\ell_\ast(\bm{\theta}).
$$
Let $\rho_n$ ($n = 1,2,\dots$) be a sequence that satisfies $\rho_n > 0$ and $\lim_{n \rightarrow \infty}\rho_n = 0$.  Let the prenet solution with $\gamma = 0$ and $\rho = \rho_n$ be $\hat{\bm{\theta}}_{\rho_n}$.  
Then we obtain
\begin{equation*}
	\lim_{n \rightarrow \infty}d(\hat{\bm{\theta}}_{\rho_n},\Theta_q^\ast) = 0\quad\text{a.s.},
\end{equation*}
where $d(\bm{a},B):=\inf_{\bm{b}\in B}d(\bm{a},\bm{b})$.
\end{prop}
\begin{proof}
	The proof is given in \ref{qmin:app}.
\end{proof}
\begin{remark}
Proposition \ref{prop_qmin} uses a set of true solutions $\Theta_q^\ast$ instead of one true solution $\bm{\theta}_q^*$.   This is because even if the quartimin solution does not have a rotational indeterminacy, it still has an indeterminacy with respect to sign and permutation of columns of the loading matrix.
\end{remark}
\begin{remark}
	In the lasso-type penalization procedure, it is interesting to investigate the consistency in model selection and asymptotic normality (e,g, \citealp{FanLi:2001}).  However, in general, it is difficult to show the model selection consistency and the asymptotic normality simultaneously \citep{knight2000asymptotics}.  Further investigation of the asymptotic properties is beyond the scope of this paper but should be considered as a future research topic.
\end{remark}

\subsection{Miscellaneous}
\subsubsection{Comparison with general rotation criterion}
With the penalization procedure, we can construct a penalty term that is based on rotation criteria other than quartimin criterion.  For example, the penalty based on the varimax rotation \citep{kaiser1958varimax} may be expressed as
\begin{equation*}
P(\bm{\Lambda}) = \sum_{j = 1}^m  \sum_{k \ne j} \sum_{i = 1}^p \lambda_{ij}^2\lambda_{ik}^2 + \frac{1}{p}  \sum_{j=1}^m  \left( \sum_{i=1}^p \lambda_{ij}^2 \right)^2.
\end{equation*}
The derivation is given in Appendix B.  Although the varimax rotation is very popular, the corresponding penalty does not have the property that $\rho \rightarrow \infty$ leads to the perfect simple structure.  In fact, $\hat{\bm{\Lambda}} = \bm{0}$ as $\rho \rightarrow \infty$.  We have derived several penalty terms based on the rotation criteria, but only the quartimin criterion possesses the perfect simple structure when $\rho \rightarrow \infty$.

\subsubsection{Normalization of factor loadings}
In factor rotation, the normalized loading matrix
\begin{equation*}
	\tilde{\lambda}_{ij} = \frac{\lambda_{ij}}{\sqrt{\sum_{k = 1}^m \lambda_{ik}^2}}
\end{equation*}
often provides better results than the unnormalized loading matrix.  In the prenet penalization, we may use the normalized penalty, in which $\lambda_{ij}$ is replaced with $\tilde{\lambda}_{ij}$
\begin{equation*}
P(\bm{\Lambda}) = \sum_{i = 1}^p \sum_{j = 1}^{m-1 } \sum_{k > j} \left\{ \frac{1}{2} (1-\gamma)  \frac{\lambda_{ij}^2\lambda_{ik}^2}{(\sum_{q = 1}^m \lambda_{iq}^2)^2} + \gamma \frac{|\lambda_{ij}||\lambda_{ik}|}{\sum_{q = 1}^m \lambda_{iq}^2} \right\}. \label{normalized penalty}
\end{equation*}
However, the above penalty is scale-invariant, that is, $P(\bm{\Lambda}) = P(a\bm{\Lambda})$ for any $a \in (0,1]$.  This fact is completely opposed to the basic concept of the penalization procedure that the penalty term should be small when the elements of $\bm{\Lambda}$ are small.  Therefore, the normalized prenet penalty does not make any sense.  
Instead, we may use a weighted penalty
\begin{equation}
P(\bm{\Lambda}) = \sum_{i=1}^p \sum_{j = 1}^{m-1 } \sum_{k > j}  \left\{ \frac{1}{2} (1-\gamma) w_i^2 \lambda_{ij}^2\lambda_{ik}^2 + \gamma w_i |\lambda_{ij}||\lambda_{ik}| \right\}, \label{weighted prenet penalty}
\end{equation}
where $w_i = 1/\sum_{q = 1}^m \hat{\lambda}_{iq}^2$.  Here, $\hat{\lambda}_{iq}$ is the $(i,q)$th element of the maximum likelihood estimate of loading matrix $\hat{\bm{\Lambda}}_{\rm ML}$.  Note that $w_i$ is independent of the factor rotation.  We can show that the weighted prenet penalty in (\ref{weighted prenet penalty}) is a generalization of the quartimin criterion with the weighted loading matrix: with $\gamma = 0$ and $\rho \rightarrow 0$, we obtain a normalized loading matrix estimated by the quartimin criterion.  This property can be proved in the same manner as Proposition \ref{prop_qmin}.

\section{Algorithm}
It is well-known that the solutions estimated by the lasso-type penalization methods are not usually expressed in a closed form, because the penalty term includes an indifferentiable function.  As the objective function of the prenet is nonconvex and nonseparable, it is not easy to construct an efficient algorithm to obtain a global minimum.  Here, we use the GEM algorithm, in which the latent factors are considered to be missing values.  The complete-data log-likelihood function is increased with the use of the coordinate descent algorithm \citep{Friedmanetal:2010}, which is a commonly used algorithm in the lasso-type penalization.  Although our proposed algorithm is not guaranteed to attain the global minimum, our algorithm decreases the objective function at each step. 

The prenet tends to be multimodal for large $\rho$, because our algorithm is a generalization of the $k$-means algorithm (the $k$-means algorithm also depends on the initial values).  Therefore, we prepare many initial values, estimate the solutions for each initial value, and select a solution that minimizes the penalized loss function.  In this case, it seems that we require heavy computational loads.  However, as described in Subsection \ref{sec:initial value}, we can construct an efficient algorithm for a sufficiently large $\rho$.

\subsection{Update equation for fixed tuning parameters}
We provide update equations of factor loadings and unique variances when $\rho$ and $\gamma$ are fixed.  Suppose that $\bm{\Lambda}_{\rm old}$ and $\bm{\Psi}_{\rm old}$ are the current values of factor loadings and unique variances, respectively.  The parameter can be updated by minimizing the negative expectation of the complete-data penalized log-likelihood function with respect to $\bm{\Lambda}$  and $\bm{\Psi}$ (e.g., \citealp{hirose2015sparse}):
\begin{eqnarray}
Q(\bm{\Lambda},\bm{\Psi}) & = & \sum_{i = 1}^p \log \psi_i + \sum_{i = 1}^p \frac{s_{ii} - 2\bm{\lambda}_i^T\bm{b}_i+ \bm{\lambda}_i^T \bm{A}\bm{\lambda}_i}{\psi_i} + \rho P( \bm{\Lambda})+{\rm const.,} \label{ECL}
\end{eqnarray}
where $\bm{b}_i = \bm{M}^{-1}\bm{\Lambda}_{\rm old}^T\bm{\Psi}_{\rm old }^{-1}\bm{s}_i$ and $\bm{A} = \bm{M} ^{-1} + \bm{M}^{-1}\bm{\Lambda}_{\rm old}^T\bm{\Psi}_{\rm old }^{-1}\bm{S}\bm{\Psi}_{\rm old }^{-1}\bm{\Lambda}_{\rm old }\bm{M}^{-1}$.  Here, $\bm{M} = \bm{\Lambda}_{\rm old }^T\bm{\Psi}_{\rm old }^{-1}\bm{\Lambda}_{\rm old } + \bm{I}_m$, and $\bm{s}_i$ is the $i$th column vector of $\bm{S}$.  In practice, minimization of (\ref{ECL}) is difficult, because the prenet penalty consists of nonconvex and nonseparable functions.  Therefore, we use a coordinate descent algorithm and obtain updated parameters, say $(\bm{\Lambda}^{\rm new}, \bm{\Psi}^{\rm new})$, which decrease the negative penalized complete-data log-likelihood function
\begin{eqnarray*}
Q(\bm{\Lambda}^{\rm new},\bm{\Psi}^{\rm new}) \le  Q(\bm{\Lambda},\bm{\Psi}) .\label{maxECL}
\end{eqnarray*}
The update equation of the coordinate descent algorithm is given in Appendix C.

After updating $\bm{\Lambda}$ using the coordinate descent algorithm, the unique variances of $\bm{\Psi}$ are updated by minimizing the function (\ref{ECL})
\begin{equation*}
\psi_i^{\rm new} = s_{ii} - 2 (\bm{\lambda}^{\rm new}_i)^T\bm{b}_i  +  (\bm{\lambda}^{\rm new}_i)^T \bm{A} \bm{\lambda}^{\rm new}_i \quad  \mbox{for $i = 1,\dots,p$,}
\end{equation*}
where $\psi_i^{\rm new}$ is the $i$th diagonal element of $\bm{\Psi}^{\rm new}$, and $\bm{\lambda}^{\rm new}_i$ is the $i$th row of $\bm{\Lambda}^{\rm new}$.

\subsection{Efficient algorithm for sufficiently large $\rho$}\label{sec:initial value}
For sufficiently large $\rho$, the $i$th column of loading matrix $\bm{\Lambda}$ has at most one nonzero element, denoted by $\lambda_{ij}$.  With the expectation--maximization (EM) algorithm, we can easily find the location of the nonzero parameter when the current value of the parameter is given.  Assume that the $(i,j)$th element of the loading matrix is nonzero and the $(i,k)$th elements ($k \neq j$) are zero.  Because the penalty function attains zero for sufficiently large $\rho$, it is sufficient to minimize the following function:
\begin{eqnarray}
f(\lambda_{ij}) = \bm{\lambda}_i^T \bm{A}\bm{\lambda}_i - 2\bm{\lambda}_i^T\bm{b}_i = a_{jj}\lambda_{ij}^2 - 2\lambda_{ij}b_{ij} \label{minimization function:initial}
\end{eqnarray}
The minimizer is easily obtained by
\begin{equation}
\hat{\lambda}_{ij} = {b_{ij}}/{a_{jj}}. \label{eq:minimizerlam}	
\end{equation}
Substituting (\ref{eq:minimizerlam}) into (\ref{minimization function:initial}) gives us
$f(\hat{\lambda}_{ij}) = -\frac{b_{ij}^2}{a_{jj}}$.  Therefore, the index $j$ that minimizes the function $f(\lambda_{ij})$ is given by
\begin{equation*}
j = {\rm argmax}_k\frac{b_{ik}^2}{a_{kk}},
\end{equation*}
and $\bm{\lambda}_i$ is updated as $\hat{\lambda}_{ij} = {b_{ij}}/{a_{jj}}$ and $\hat{\lambda}_{ik} = 0$ $(\forall k \ne j)$.

\subsection{Selection of the maximum value of $\rho$}
The value of $\rho_{\max}$, which is the minimum value of $\rho$ that produces the perfect simple structure, is easily obtained using $\hat{\bm{\Lambda}}$ given by (\ref{eq:minimizerlam}).  Assume that $\hat{\lambda}_{ij}\neq 0$ and $\hat{\lambda}_{ik} = 0$ ($k \neq j$).  Using the update equation of $\lambda_{ik}$ in (\ref{lambdaupdate}) and the soft thresholding function in (\ref{uelasso}), we show that the regularization parameter $\rho$ must satisfy the following inequality to ensure that $\lambda_{ik}$ is estimated to be zero:
\begin{equation*}
\left|\frac{b_{ik}- a_{kj}\hat{\lambda}_{ij}}{a_{kk}+\rho\psi_i(1-\gamma)\hat{\lambda}_{ij}^2 }\right| \le \frac{\psi_i}{a_{kk}+\rho\psi_i(1-\gamma)\hat{\lambda}_{ij}^2 }\rho \gamma|\hat{\lambda}_{ij}|.
\end{equation*}
Thus, the value of $\rho_{\max}$ is given by
\begin{equation*}
\rho_{\max} = \max_{i} \max_{k \in C_i}\frac{|b_{ik}- a_{kj}\hat{\lambda}_{ij}|}{\gamma\psi_i|\hat{\lambda}_{ij}|}, \label{eq:rho_max}
\end{equation*}
where $C_i = \{k|k \ne j, \hat{\lambda}_{ij} \ne 0\}$.

\subsection{Estimation of the entire path of solutions}
The entire path of solutions can be produced with the grid of increasing values $\{\rho_1,\dots,\rho_K\}$.  Here, $\rho_K$ is given by (\ref{eq:rho_max}), and $\rho_{1} = \rho_{K} \Delta \sqrt{\gamma}$, where $\Delta$ is a small value such as $0.001$.  The term $\sqrt{\gamma}$ allows us to estimate a variety of models even if $\gamma$ is small.

The entire solution path can be made using a decreasing sequence $\{\rho_K\dots,\rho_1\}$, starting with $\rho_K$.  Note that the proposed algorithm at $\rho_K$ does not always converge to the global minimum, so that we prepare many initial values, estimate solutions for each initial value with the use of the efficient algorithm described in Subsection \ref{sec:initial value}, and select a solution that minimizes the penalized log-likelihood function.  We can use the warm start, which can provide the starting values of the parameters: the solution at $\rho_{k-1}$ can be computed using the solution at $\rho_k$, which leads to improved and smoother objective value surfaces \citep{Mazumderetal:2009}.  The cold start may be used, but it requires heavy computational loads.

\section{Monte Carlo simulations}\label{sec:simulation}
In this simulation study, we use four simulation models.  The first three models are as below.
\begin{flushleft}
{\bf Model (A):}
\end{flushleft}
\vspace{-2em}
\begin{eqnarray*}
\bm{\Lambda} &=&
\left(
\begin{array}{rrrrrr}
0.95&0.9&0.85&0.0&0.0&0.0\\
0.0&0.0&0.0&0.8&0.75&0.7
\end{array}
\right)^T,
\end{eqnarray*}
\begin{flushleft}
{\bf Model (B):}
\end{flushleft}
\vspace{-1em}
 \begin{eqnarray*}
 \bm{\Lambda} &=&
 \left(
 \begin{array}{rrrrrr}
0.9&0.8&0.7&0.2&0.2&0.2\\
0.2&0.2&0.2&0.9&0.8&0.7
 \end{array}
 \right)^T,
\end{eqnarray*}
\begin{flushleft}
 {\bf Model (C):}
\end{flushleft}
\vspace{-1em}
 \begin{eqnarray*}
 \bm{\Lambda} &=&
 \left(
 \begin{array}{rrrrrr}
0.8 \cdot \bm{1}_{25} &\bm{0}_{25}&\bm{0}_{25}&\bm{0}_{25}\\
\bm{0}_{25} &0.75\cdot \bm{1}_{25}&\bm{0}_{25}&\bm{0}_{25}\\
\bm{0}_{25}&  \bm{0}_{25} &0.7\cdot\bm{1}_{25}&\bm{0}_{25}\\
\bm{0}_{25} & \bm{0}_{25} &\bm{0}_{25}&0.65\cdot\bm{1}_{25}\\
 \end{array}
 \right),
 \end{eqnarray*}
where $\bm{1}_{25}$ is a $25$-dimensional vector with each element being 1, and $\bm{0}_{25}$ is a 25-dimensional zero vector.  We also use Model (D), which is similar to Model (C) but replace 100 randomly chosen elements out of 300 zero elements of $\bm{\Lambda}$ with $U(0.4,0.6)$.  If the communality of $\bm{\Lambda}$ is greater than 1, the corresponding row is scaled so that the communality becomes 0.95.  Then, the unique variances are obtained by $\bm{\Psi} = {\rm diag} (\bm{I} -  \bm{\Lambda}\bm{\Lambda}^T)$.

In Models (A) and (C), the loading matrix possesses the perfect simple structure.  Model (C) is a large model compared with Model (A).  The loading matrix of Model (B) is not sparse but we can interpret that the first factor is related to the first three observed variables, and the second factor is related to the remaining three observed variables.  As the loading matrix is the same as that given in Section \ref{sec:illustration}, the prenet penalty is expected to outperform the lasso.  Model (D) is as large as Model (C) but does not possess the perfect simple structure.  We use Model (D) to explore the performance of the proposed procedure when the true loading matrix does not possess the perfect simple structure.

The model parameter is estimated by the prenet penalty using $\gamma = 1.0$ and $\gamma = 0.01$, and the minimax concave penalty (MC penalty; \citealp{Zhang:2010})
\begin{eqnarray*}
\rho P(\bm{\Lambda};\rho;\gamma)&=& \sum_{i = 1}^p\sum_{j = 1}^m\rho\int_0^{|\lambda_{ij}|}\left(1-\frac{x}{\rho\gamma}\right)_+dx\\
& = &\sum_{i = 1}^p\sum_{j = 1}^m\left\{ \rho \left(|\lambda_{ij}|-\frac{\lambda_{ij}^2}{2\rho\gamma}\right) I(|\lambda_{ij}| < \rho\gamma)+ \frac{\rho^2\gamma}{2}I(|\lambda_{ij}| \ge \rho\gamma) \right\},
\end{eqnarray*}
with $\gamma = \infty$ and $\gamma = 3$.  Note that $\gamma = \infty$ with the MC penalty is equivalent to the lasso.  The regularization parameter is selected by the Akaike information criterion (AIC), Bayesian information crietrion (BIC), and extended BIC (EBIC; \citealp{chen2008extended})
\begin{eqnarray*}
{\rm AIC} & = & -2 \ell(\hat{\bm{\Lambda}},\hat{\bm{\Psi}}) + 2p_0,\\
{\rm BIC} & = & -2 \ell(\hat{\bm{\Lambda}},\hat{\bm{\Psi}}) + (\log N ) p_0,\\
{\rm EBIC} & = & -2 \ell(\hat{\bm{\Lambda}},\hat{\bm{\Psi}}) +  (\log N ) p_0  + 2p_0 \delta \log(pm),
\end{eqnarray*}
where $p_0$ is the number of nonzero parameters, and $\delta \in [0,1]$ is a hyper-parameter of the prior distribution of the EBIC.  In this simulation, we select $\delta = 1$.  For each model, $T = 100$ data sets are generated with $\bm{x} \sim N(\bm{0},\bm{\Lambda}\bm{\Lambda}^T + \bm{\Psi})$.  The number of observations is $n = 50, 100$, and $500$.  Tables \ref{table:simulation1}--\ref{table:simulation4} show the mean squared error defined by
\begin{eqnarray*}
{\rm MSE} = \frac{1}{T} \sum_{s = 1}^{T} \frac{\| \bm{\Lambda} - \hat{\bm{\Lambda}}^{(s)} \|^2}{pm},
\end{eqnarray*}
where $\hat{ \bm{\Lambda}}^{(s)}$ is the estimate of the loading matrix using the $s$th dataset.  We also compare the true positive rate (TPR) and false positive rate (FPR) of the loading matrix over 100 simulations.

\begin{table}
\caption{\label{table:simulation1}Mean squared errors, true positive rates, and false positive rates of estimated factor loadings for Model (A).} 
\centering
\begin{tabular}{llrrrrrrrrrr}
  \hline
 && \multicolumn{3}{c}{$n = 50$} & \multicolumn{3}{c}{$n = 100$}&\multicolumn{3}{c}{$n = 500$}  \\
 && MSE&TPR&FPR & MSE&TPR&FPR & MSE&TPR&FPR    \\
  \hline
AIC & lasso           & 0.10 & 1.00 & 0.56 & 0.04 & 1.00 & 0.55 & 0.01 & 1.00 & 0.55 \\
    & MC               & 0.07 & 1.00 & 0.24 & 0.02 & 1.00 & 0.14 & 0.00 & 1.00 & 0.14 \\
    & prenet$_1$       & 0.05 & 1.00 & 0.14 & 0.02 & 1.00 & 0.12 & 0.00 & 1.00 & 0.12 \\
    & prenet$_{.01}$   & 0.04 & 1.00 & 0.06 & 0.02 & 1.00 & 0.06 & 0.00 & 1.00 & 0.06 \\
BIC & lasso            & 0.11 & 1.00 & 0.47 & 0.06 & 1.00 & 0.38 & 0.01 & 1.00 & 0.36 \\
    & MC               & 0.07 & 1.00 & 0.17 & 0.02 & 1.00 & 0.07 & 0.00 & 1.00 & 0.00 \\
    & prenet$_1$       & 0.04 & 1.00 & 0.04 & 0.01 & 1.00 & 0.01 & 0.00 & 1.00 & 0.00 \\
    & prenet$_{.01}$   & 0.03 & 1.00 & 0.01 & 0.01 & 1.00 & 0.00 & 0.00 & 1.00 & 0.00 \\
EBIC & lasso            & 0.59 & 0.84 & 0.21 & 0.11 & 1.00 & 0.22 & 0.03 & 1.00 & 0.22 \\
    & MC               & 0.32 & 0.92 & 0.11 & 0.04 & 1.00 & 0.06 & 0.01 & 1.00 & 0.00 \\
    & prenet$_1$       & 0.03 & 1.00 & 0.00 & 0.01 & 1.00 & 0.00 & 0.00 & 1.00 & 0.00 \\
    & prenet$_{.01}$   & 0.03 & 1.00 & 0.00 & 0.01 & 1.00 & 0.00 & 0.00 & 1.00 & 0.00 \\
       \hline
\end{tabular}
\end{table}

\begin{table}
\caption{\label{table:simulation2}Mean squared errors, true positive rates, and false positive rates of estimated factor loadings for Model (B).} 
\centering
\begin{tabular}{llrrrrrrrrrr}
  \hline
 && \multicolumn{3}{c}{$n = 50$} & \multicolumn{3}{c}{$n = 100$}&\multicolumn{3}{c}{$n = 500$}  \\
 && MSE&TPR&FPR & MSE&TPR&FPR & MSE&TPR&FPR    \\
  \hline
AIC & lasso           & 0.26 & 0.88 & --- & 0.17 & 0.90 & --- & 0.16 & 0.90 & --- \\
    & MC               & 0.30 & 0.76 & --- & 0.22 & 0.80 & --- & 0.20 & 0.80 & --- \\
    & prenet$_1$       & 0.27 & 0.77 & --- & 0.16 & 0.88 & --- & 0.16 & 0.90 & --- \\
    & prenet$_{.01}$   & 0.23 & 0.82 & --- & 0.05 & 0.98 & --- & 0.01 & 1.00 & --- \\
BIC & lasso            & 0.27 & 0.83 & --- & 0.16 & 0.88 & --- & 0.15 & 0.89 & --- \\
    & MC               & 0.30 & 0.70 & --- & 0.24 & 0.72 & --- & 0.20 & 0.77 & --- \\
    & prenet$_1$       & 0.28 & 0.63 & --- & 0.20 & 0.70 & --- & 0.15 & 0.88 & --- \\
    & prenet$_{.01}$   & 0.31 & 0.54 & --- & 0.17 & 0.65 & --- & 0.01 & 1.00 & --- \\
EBIC& lasso            & 0.35 & 0.78 & --- & 0.18 & 0.85 & --- & 0.15 & 0.88 & --- \\
    & MC               & 0.30 & 0.67 & --- & 0.24 & 0.70 & --- & 0.20 & 0.77 & --- \\
    & prenet$_1$       & 0.25 & 0.51 & --- & 0.22 & 0.52 & --- & 0.15 & 0.86 & --- \\
    & prenet$_{.01}$   & 0.25 & 0.50 & --- & 0.21 & 0.50 & --- & 0.02 & 0.98 & --- \\        \hline
    \end{tabular}
\end{table}

\begin{table}
\caption{\label{table:simulation3}Mean squared errors, true positive rates, and false positive rates of estimated factor loadings for Model (C).} 
\centering
\begin{tabular}{llrrrrrrrrrr}
  \hline
 && \multicolumn{3}{c}{$n = 50$} & \multicolumn{3}{c}{$n = 100$}&\multicolumn{3}{c}{$n = 500$}  \\
 && MSE&TPR&FPR & MSE&TPR&FPR & MSE&TPR&FPR    \\
  \hline
AIC & lasso           & 0.14 & 1.00 & 0.85 & 0.07 & 1.00 & 0.85 & 0.02 & 1.00 & 0.85 \\
    & MC               & 0.06 & 1.00 & 0.43 & 0.02 & 1.00 & 0.21 & 0.00 & 1.00 & 0.08 \\
    & prenet$_1$       & 0.02 & 1.00 & 0.04 & 0.01 & 1.00 & 0.02 & 0.00 & 1.00 & 0.03 \\
    & prenet$_{.01}$   & 0.01 & 1.00 & 0.00 & 0.01 & 1.00 & 0.00 & 0.00 & 1.00 & 0.00 \\
BIC & lasso            & 0.35 & 1.00 & 0.52 & 0.24 & 1.00 & 0.51 & 0.09 & 1.00 & 0.52 \\
    & MC               & 0.07 & 1.00 & 0.41 & 0.02 & 1.00 & 0.19 & 0.00 & 1.00 & 0.00 \\
    & prenet$_1$       & 0.01 & 1.00 & 0.00 & 0.01 & 1.00 & 0.00 & 0.00 & 1.00 & 0.00 \\
    & prenet$_{.01}$   & 0.01 & 1.00 & 0.00 & 0.01 & 1.00 & 0.00 & 0.00 & 1.00 & 0.00 \\
EBIC & lasso            & 0.91 & 0.49 & 0.06 & 0.48 & 0.98 & 0.13 & 0.22 & 1.00 & 0.16 \\
    & MC               & 0.91 & 0.52 & 0.03 & 0.50 & 0.99 & 0.04 & 0.00 & 1.00 & 0.00 \\
    & prenet$_1$       & 0.01 & 1.00 & 0.00 & 0.01 & 1.00 & 0.00 & 0.00 & 1.00 & 0.00 \\
    & prenet$_{.01}$   & 0.01 & 1.00 & 0.00 & 0.01 & 1.00 & 0.00 & 0.00 & 1.00 & 0.00 \\
       \hline
\end{tabular}
\end{table}

\begin{table}
\caption{\label{table:simulation4}Mean squared errors, true positive rates, and false positive rates of estimated factor loadings for Model (D).} 
\centering
\begin{tabular}{llrrrrrrrrrr}
  \hline
 && \multicolumn{3}{c}{$n = 50$} & \multicolumn{3}{c}{$n = 100$}&\multicolumn{3}{c}{$n = 500$}  \\
 && MSE&TPR&FPR & MSE&TPR&FPR & MSE&TPR&FPR    \\
  \hline
AIC & lasso           & 0.28 & 1.00 & 0.93 & 0.27 & 1.00 & 0.94 & 0.25 & 1.00 & 0.94 \\
    & MC               & 0.17 & 0.99 & 0.63 & 0.13 & 0.99 & 0.51 & 0.05 & 1.00 & 0.19 \\
    & prenet$_1$       & 0.27 & 1.00 & 0.92 & 0.24 & 1.00 & 0.92 & 0.43 & 1.00 & 0.92 \\
    & prenet$_{.01}$   & 0.66 & 1.00 & 0.99 & 0.61 & 1.00 & 0.99 & 0.50 & 1.00 & 1.00 \\
BIC & lasso            & 0.32 & 0.99 & 0.91 & 0.29 & 1.00 & 0.91 & 0.25 & 1.00 & 0.91 \\
    & MC               & 0.19 & 0.99 & 0.62 & 0.11 & 0.99 & 0.48 & 0.05 & 1.00 & 0.15 \\
    & prenet$_1$       & 0.35 & 0.99 & 0.88 & 0.30 & 0.99 & 0.87 & 0.44 & 0.99 & 0.87 \\
    & prenet$_{.01}$   & 0.66 & 1.00 & 0.99 & 0.60 & 1.00 & 0.99 & 0.50 & 1.00 & 1.00 \\
EBIC & lasso            & 0.84 & 0.97 & 0.63 & 0.54 & 0.99 & 0.75 & 0.24 & 1.00 & 0.83 \\
    & MC               & 0.22 & 0.99 & 0.61 & 0.12 & 0.99 & 0.47 & 0.03 & 1.00 & 0.13 \\
    & prenet$_1$       & 1.31 & 0.43 & 0.10 & 0.66 & 0.97 & 0.56 & 0.29 & 0.99 & 0.71 \\
    & prenet$_{.01}$   & 1.32 & 0.41 & 0.09 & 0.62 & 1.00 & 0.98 & 0.50 & 1.00 & 0.99 \\
       \hline
\end{tabular}
\end{table}

We obtain the following empirical observations for each simulation model:
\begin{description}
\item[Model (A):] In almost all cases, the prenet penalty outperforms the lasso and MC in terms of both MSE and TPR.  For example, when $n =50$, the EBIC based on lasso and MC tends to select too simple models;  the estimated model is often one-factor model, which is completely different from the true loading matrix.   For the prenet penalty, the EBIC may select simple models (like the lasso), but it performs very well.  This is because the prenet penalty estimates a model that possesses the perfect simple structure for large $\rho$.
\item[Model (B):] The prenet with $\gamma = 0.01$ outperforms the other methods, as seen in Section \ref{sec:illustration}.  In particular, when $n =500$, the prenet with $\gamma = 0.01$ performs very well irrespective of the model selection criteria.
\item[Model (C):] The result is similar to that of Model (A).  With high-dimensional data, the MC tends to perform much better than the lasso.  The performance of the prenet penalty is almost independent of $\gamma$.
\item[Model (D):] The prenet penalty performs worse than the lasso-type regularization, because the true loading matrix is far from the perfect simple structure.  In particular, when $\gamma = 0.01$, the prenet performs poorly. \end{description}

\section{Real data analyses}
\subsection{Big five personality traits}\label{sec.big5}
The first example is the survey data regarding the big five personality traits collected from Open Source Psychometrics Project \citep{open2011}.  8582 responders in the US region are asked to assess their own personality based on 50 questions developed by \citet{goldberg1992development}.  Each question asks how well it describes the statement of the responders on a scale of 1--5.   It is well-known that the personality is characterized by five common factors: openness to experience, conscientiousness, extraversion, agreeableness, and neuroticism.  We investigate whether these five personality traits can be properly extracted by using the prenet penalization. 

First, we apply the prenet penalization and the varimax rotation with maximum likelihood estimate, and compare the loading matrices estimated by these two methods.  With the prenet penalization, we choose tuning parameters which achieve the perfect simple structure ($\lambda= 0.74$, $\gamma=1.0$).  The heatmap of the loading matrices are shown in Figure \ref{fig:big5_1}.   
\begin{figure}[!t]
\centering
\includegraphics[height=7.0cm]{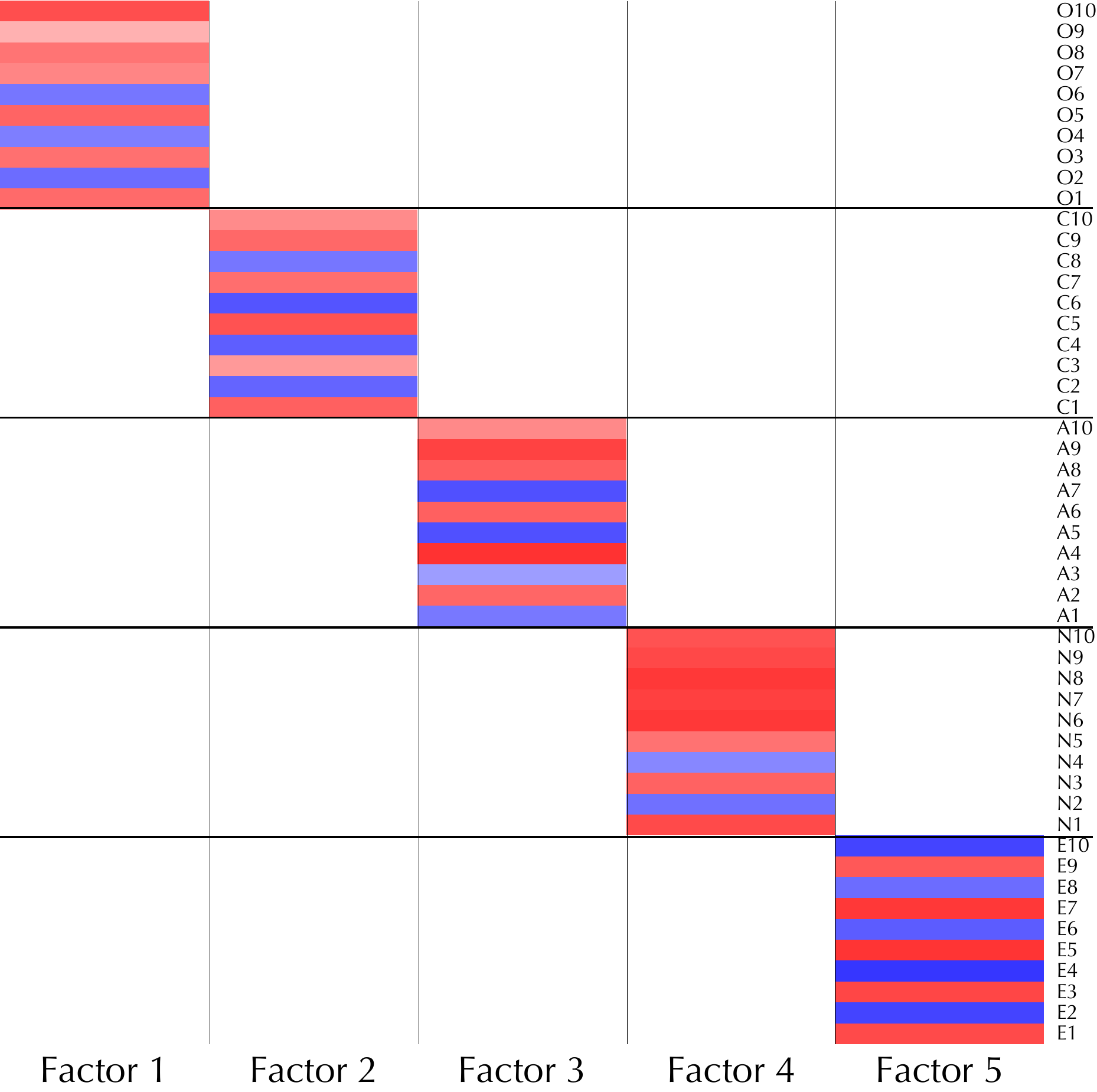}\hspace{5mm}
\includegraphics[height=7.0cm]{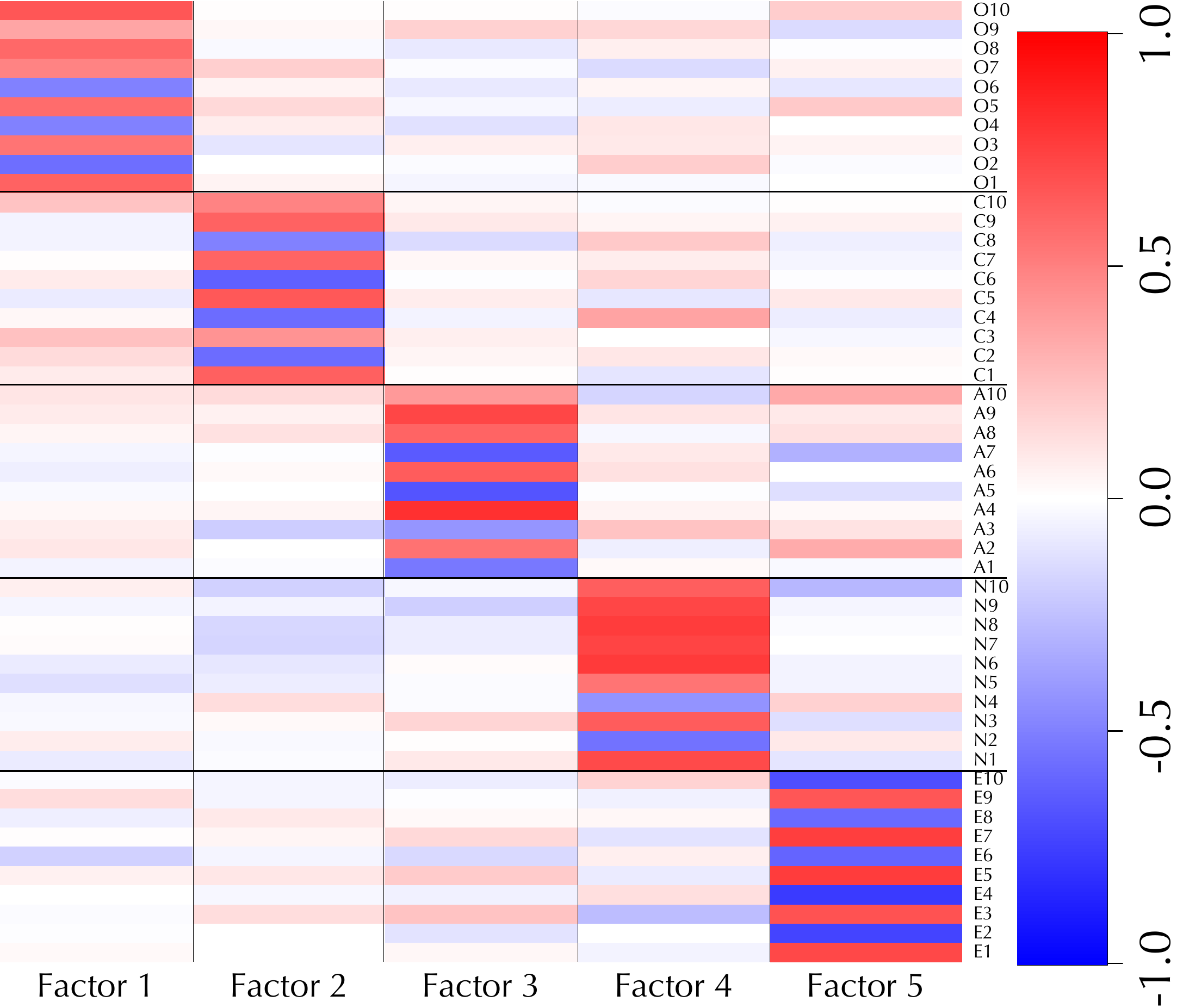}\hspace{5mm}
  \caption{Heatmap of the loading matrices on big five personality traits data.  The left panel corresponds to the prenet penalization with $\lambda=0.74$ and $\gamma=1.0$, and the right panel corresponds to the varimax rotation.  Each cell corresponds to the factor loading, and the depth of color indicates the magnitude of the value of the factor loading.}	 
\label{fig:big5_1}
\end{figure}

The result of Figure \ref{fig:big5_1} shows that the prenet penalization is able to produce a sufficiently sparse loading matrix which allows a clear interpretation of the five personality traits.  A loading matrix estimated by the varimax rotation is not sufficiently sparse but can be appropriately interpreted.  To investigate how well the estimated models are fitted to data, the values of goodness-of-fit (GOF) indices are compared.  The results are 
SRMR = 0.110, RMSEA = 0.241, and CFI = 0.733 for the prenet penalization, and SRMR = 0.032, RMSEA = 0.105, and CFI = 0.846 
for the varimax rotation.  Indicators of good model fits are SRMR $\leq$ 0.05, RMSEA $\leq$ 0.08, and CFI $\geq$ 0.90 \citep{hu1999cutoff}.  The GOF indices of varimax rotation are better than those of prenet penalization.

However, it is seen that the prenet penalization performs relatively well in terms of prediction of future data and interpretation of five personality traits. Figure \ref{fig:big5_boxplot} depicts boxplots of negative log-likelihood value $\ell_{\rm ML}(\bm{\Lambda},\bm{\Psi}) $ in (\ref{taisuuyuudo}) (left panel) and degrees of sparsity (i.e., proportion of nonzero values, right panel) for $n$ random subsampled data with $n=100$, $n=200$, $n=500$, and $n=1000$.   Tuning parameters in the prenet penalty are selected by the BIC.  The boxplots are constructed by 100 simulations based on the subsampling.  To calculate the value of negative log-likelihood, the subsampled data are split into a training set and a test set; the parameter estimation is done by training data and the negative log-likelihood value is calculated with test data.  The heatmaps of mean of the loading matrices are shown in the right panel when $n=100$ and $n=1000$.  These heatmaps are depicted so that the estimated loading matrix $\hat{\bm{\Lambda}}$ is set as close to the varimax rotation with full dataset (i.e., the right panel of Figure \ref{fig:big5_1}) as possible by changing the column and the sign of column of  $\hat{\bm{\Lambda}}$.
\begin{figure}[!t]
\centering
\includegraphics[width=16.0cm]{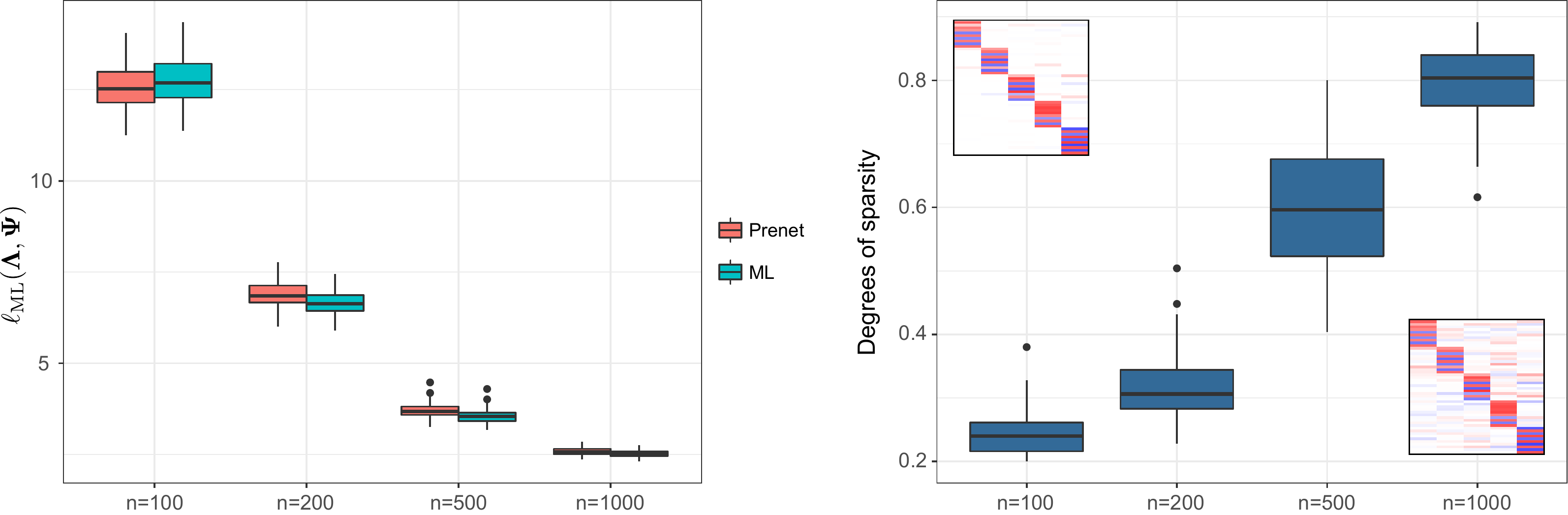}\hspace{5mm}
  \caption{Boxplots of negative log-likelihood value $\ell_{\rm ML}(\bm{\Lambda},\bm{\Psi}) $ in (\ref{taisuuyuudo}) (left panel) and degrees of sparsity (i.e., proportion of nonzero values, right panel) for $n$ random subsampled data with $n=100$, $n=200$, $n=500$, and $n=1000$.  The boxplots are constructed by 100 simulations based on the subsampling.  The heatmaps of mean of the loading matrices are depicted in the right panel when $n=100$ and $n=1000$. }	 
\label{fig:big5_boxplot}
\end{figure}

The left panel of figure \ref{fig:big5_boxplot} shows that both prenet penalization and ML result in similar values of  $\ell_{\rm ML}(\bm{\Lambda},\bm{\Psi}) $, which implies the prenet penalization is comparable to the ML.   In particular, when $n=100$, the prenet penalization slightly outperforms the ML.   The right panel of figure \ref{fig:big5_boxplot} shows that the prenet tends to produce sparse solution as $n$ becomes small.   Although the degrees of sparsity are different among subsample sizes, two heatmaps of mean of the loading matrices show that the characteristic of five personality traits is assumed to be appropriately extracted for both $n=100$ and $n=1000$.  

Figure \ref{fig:big5_2} depicts the heatmaps of the loading matrices for various values of tuning parameters on the MC penalization and the prenet penalization.  We find the tuning parameters so that the degrees of sparseness (proportion of nonzero values) of the loading matrix are approximately 20\%, 25\%, 40\%, and 50\%.  For the MC penalty, we set $\gamma = \infty$ (i.e., the lasso), $5.0$, $2.0$, and $1.01$.  For prenet penalty, the values of gamma are $\gamma=1.0,$ $0.5,$ and $0.01$.  Each cell describes the elements of the factor loadings as with Figure \ref{fig:big5_1}.
\begin{figure}[!t]
\centering
\includegraphics[width=16.0cm]{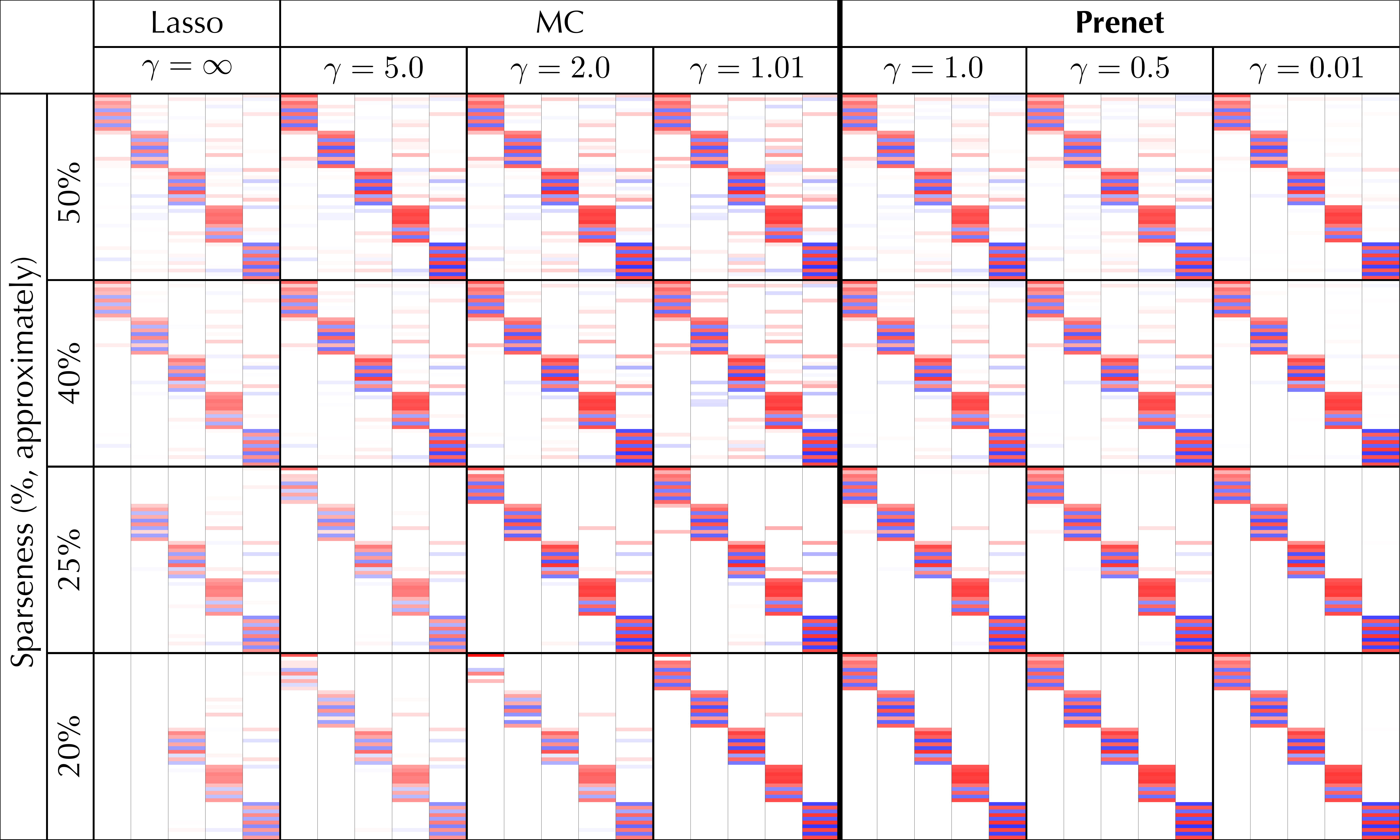}\hspace{5mm}
  \caption{Heatmaps of the loading matrices on big five personality traits data for various values of tuning parameters on the MC penalization and the prenet penalization.}	 
\label{fig:big5_2}
\end{figure}

From Figure \ref{fig:big5_2}, we obtain the empirical observations as follows:
\begin{itemize}
\item With the prenet penalization, the characteristic of five personality traits are appropriately extracted for any values of tuning parameters, which suggests that the prenet penalization is relatively robust against the tuning parameters when the loading matrix is likely to possess the perfect simple structure.  
\item The prenet penalization is able to estimate the perfect simple structure when the degree of sparseness is 20\%.  On the other hand, with the MC penalization, we are not able to estimate the perfect simple structure even when $\gamma$ is sufficiently small.
\item 	With the lasso, the number of factors becomes less than five when the degrees of sparsity are 20\% and 25\%; the five personality traits are not able to found.   When the value of $\gamma$ is not sufficiently large, the MC penalization produces five factor model.   
\item For the MC penalization, the magnitude of the absolute nonzero values becomes large as the value of $\gamma$ decreases for fixed degrees of sparsity; the MC penalization tends to increase the contrast between the zero values and nonzero values as the value of $\gamma$ becomes small.  
\end{itemize}
\subsection{Handwritten digits data}\label{sec.handwritten}
We apply the prenet penalty to well-known handwritten digits data \citep{Hastieetal:2008}.  We select the number ``0," consisting of 1194 observations with 256 pixels (variables).  The variables that have extremely small variances are removed, resulting in 184 variables.

We conduct variables clustering using the prenet, as described in Subsection \ref{sec:pss}.  To our knowledge, variables clustering of image data via factor analysis has not yet been attempted.  The prenet is compared with the $k$-means variables clustering, which is a special case of the prenet, as shown in Section \ref{sec:kmeans}.  The results for $m = 5$, 10, and 15 are depicted in Figure 2.  Color is used to denote a cluster.  When $m = 5$, we make an interesting empirical observation.  With the prenet, the same clusters show left--right symmetry, which means that we tend to write ``0" with left--right symmetry.  As the same clusters could be located in separate places, the cluster structure indicates not only the location of the pixels but also the habits of the people who usually write the letters.  On the other hand, for the $k$-means, the same clusters are located in a circle, and each cluster is characterized by the size of the circle.  The $k$-means clustering tends to assign clusters by the location of the pixels rather than people's writing habits.  Therefore, the prenet might be able to capture a more complex structure than the $k$-means.  When the number of factors (clusters) is large, the prenet and $k$-means produce similar results.
 \begin{figure}[!t]
\centering
\includegraphics[width=4.0cm]{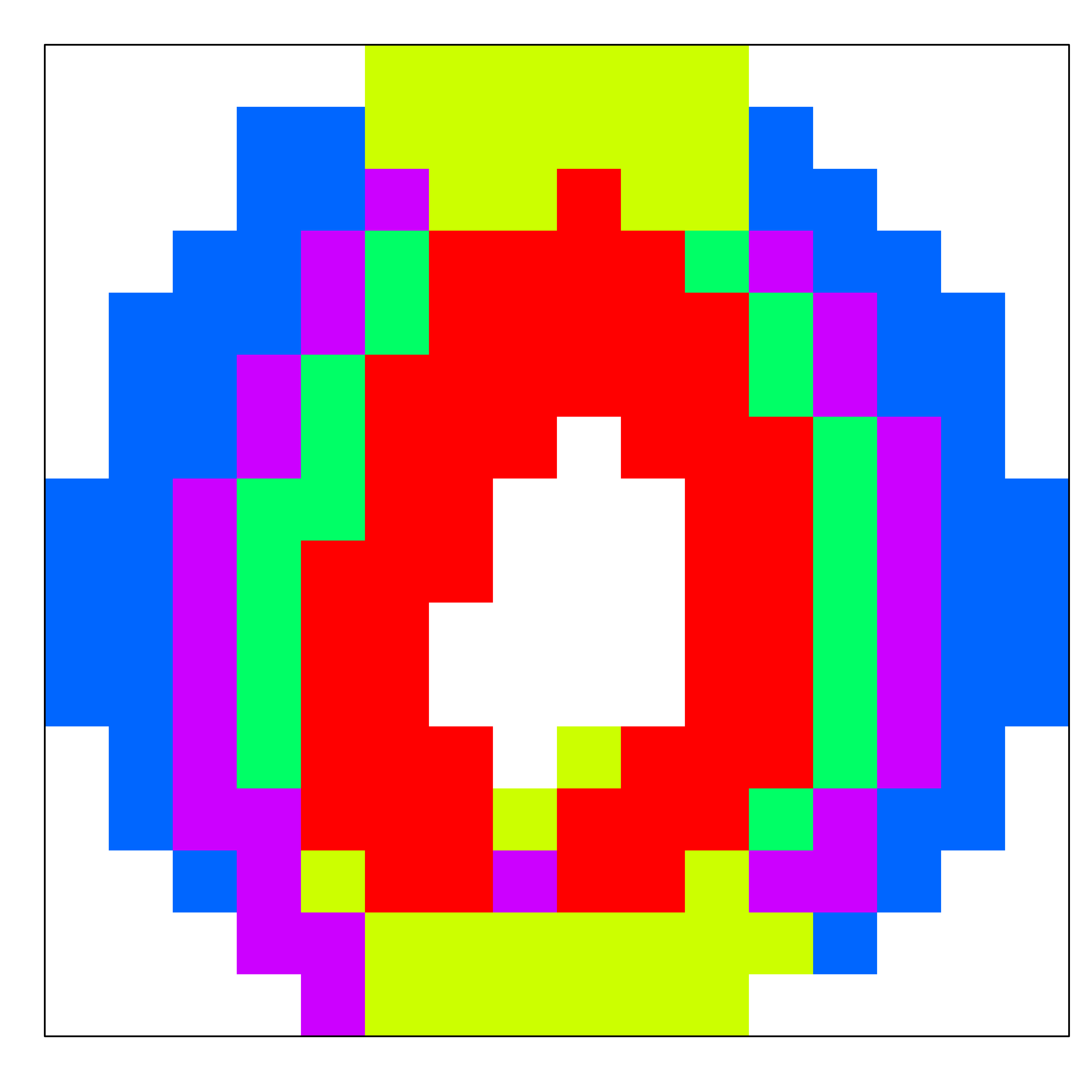}\hspace{5mm}
\includegraphics[width=4.0cm]{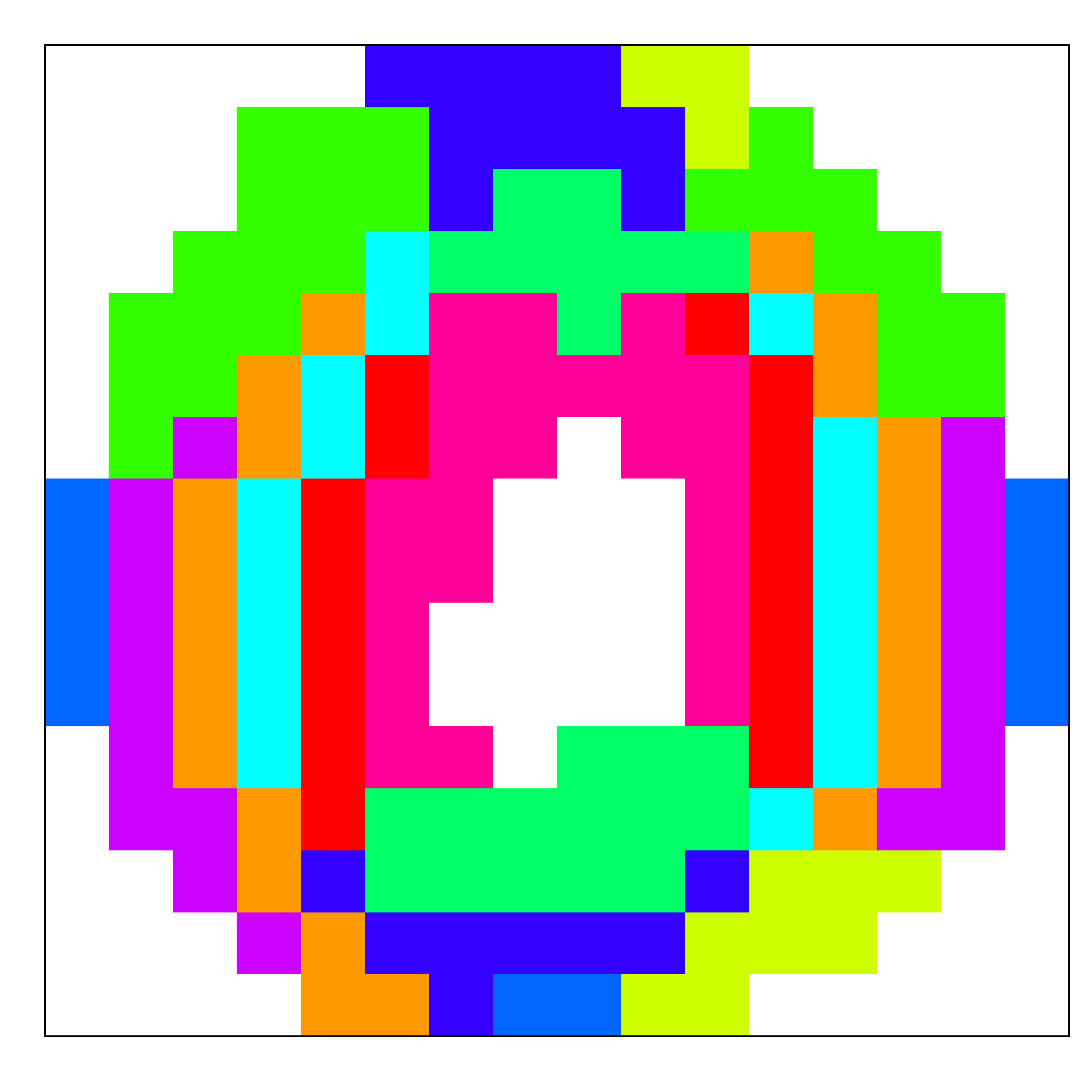}\hspace{5mm}
\includegraphics[width=4.0cm]{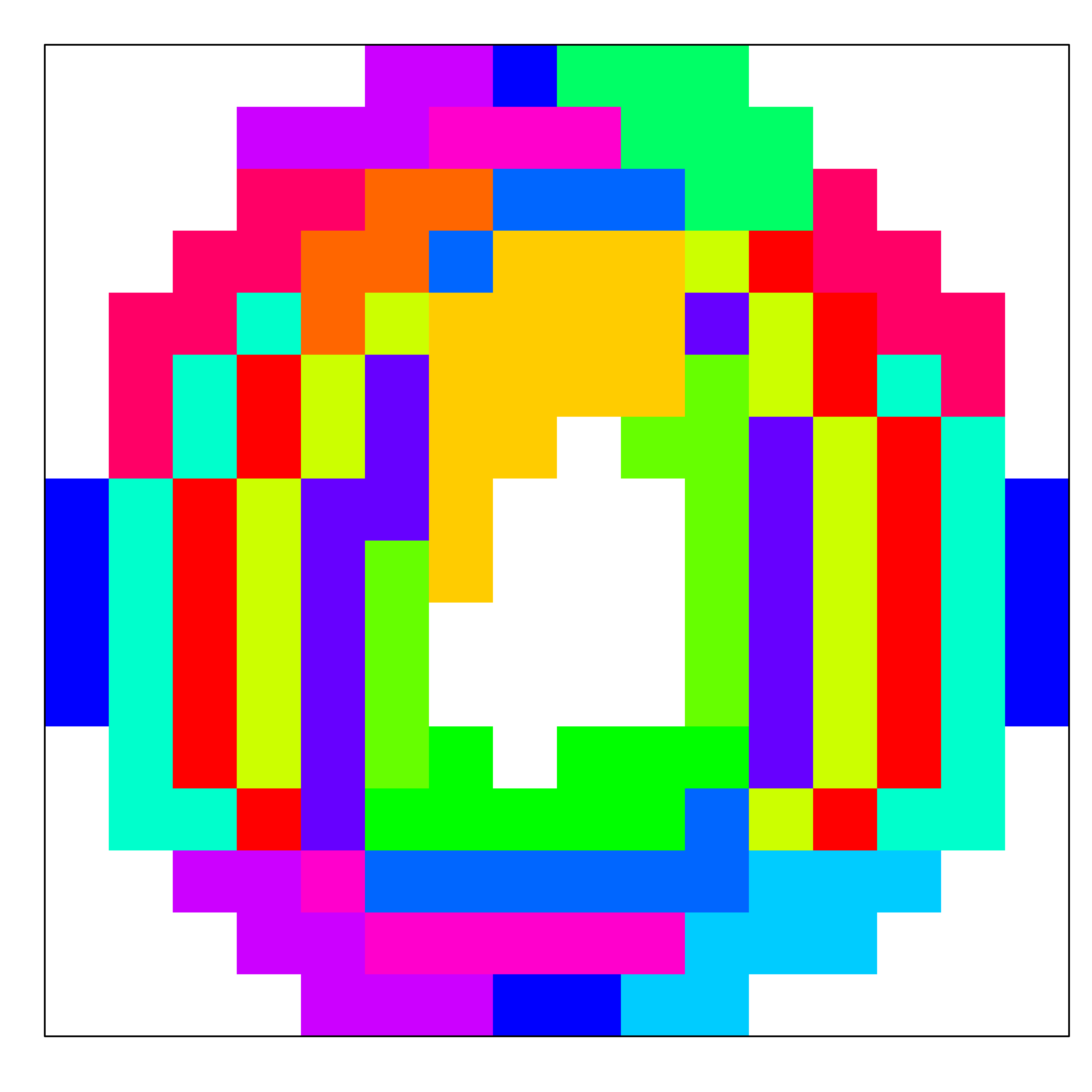}\\
\includegraphics[width=4.0cm]{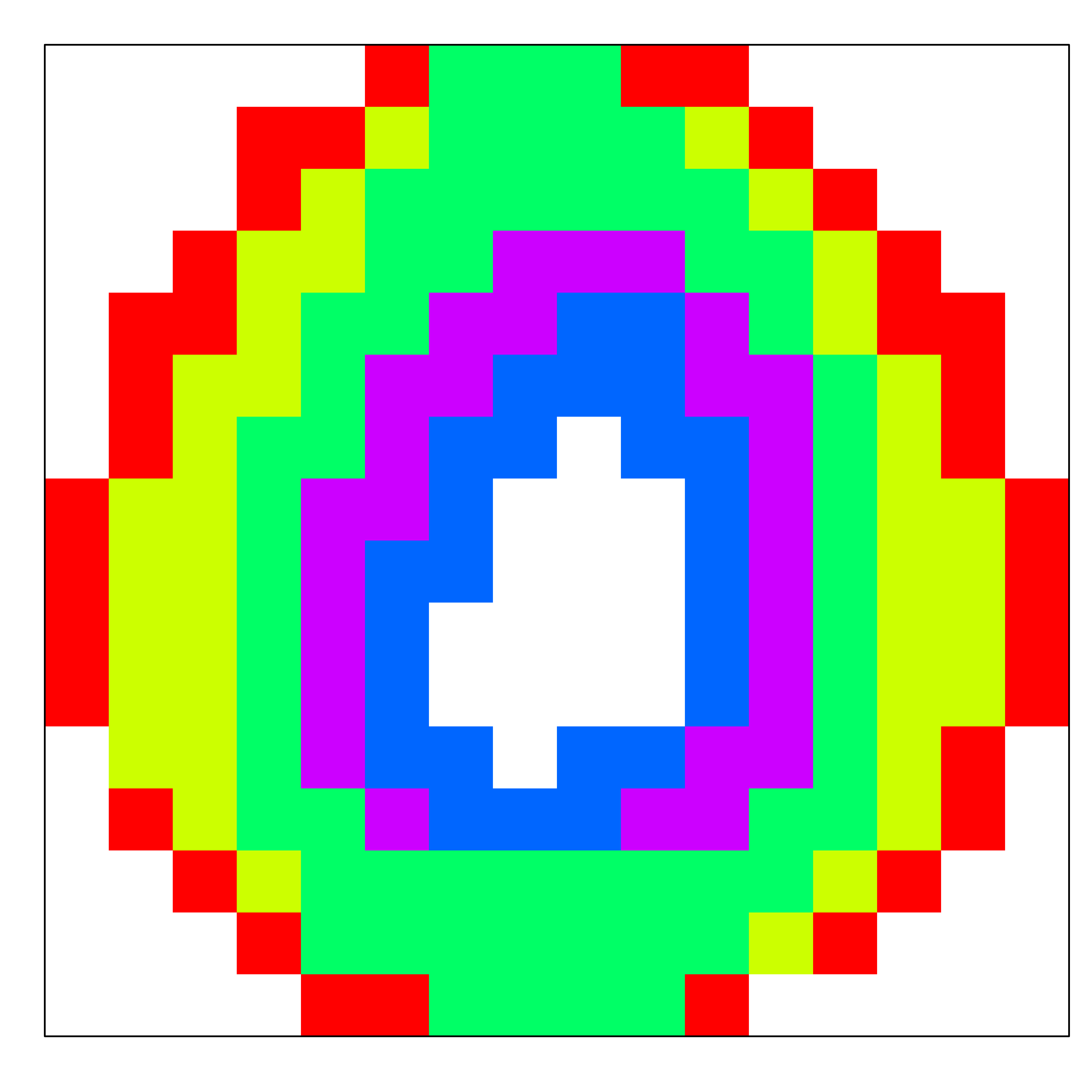}\hspace{5mm}
\includegraphics[width=4.0cm]{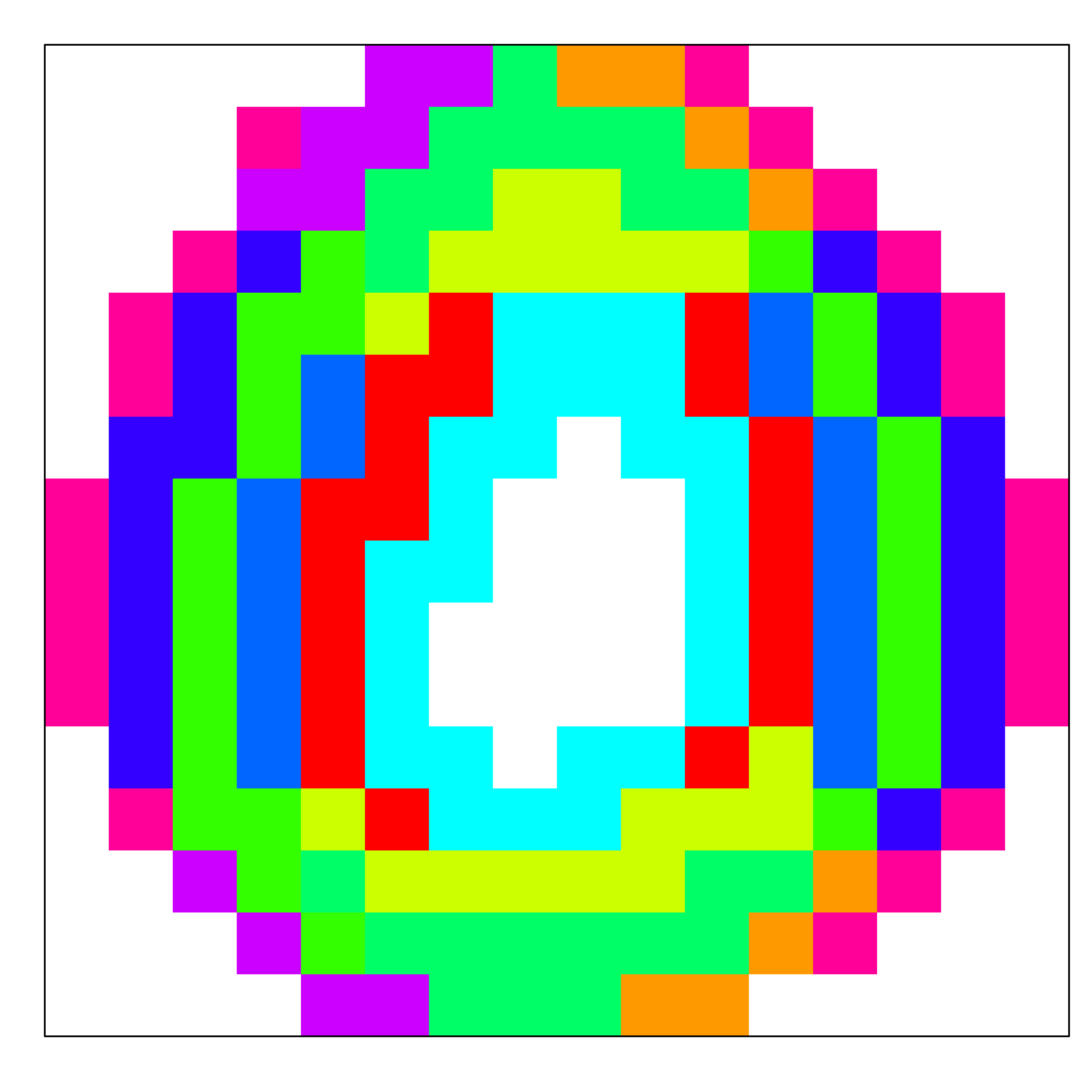}\hspace{5mm}
\includegraphics[width=4.0cm]{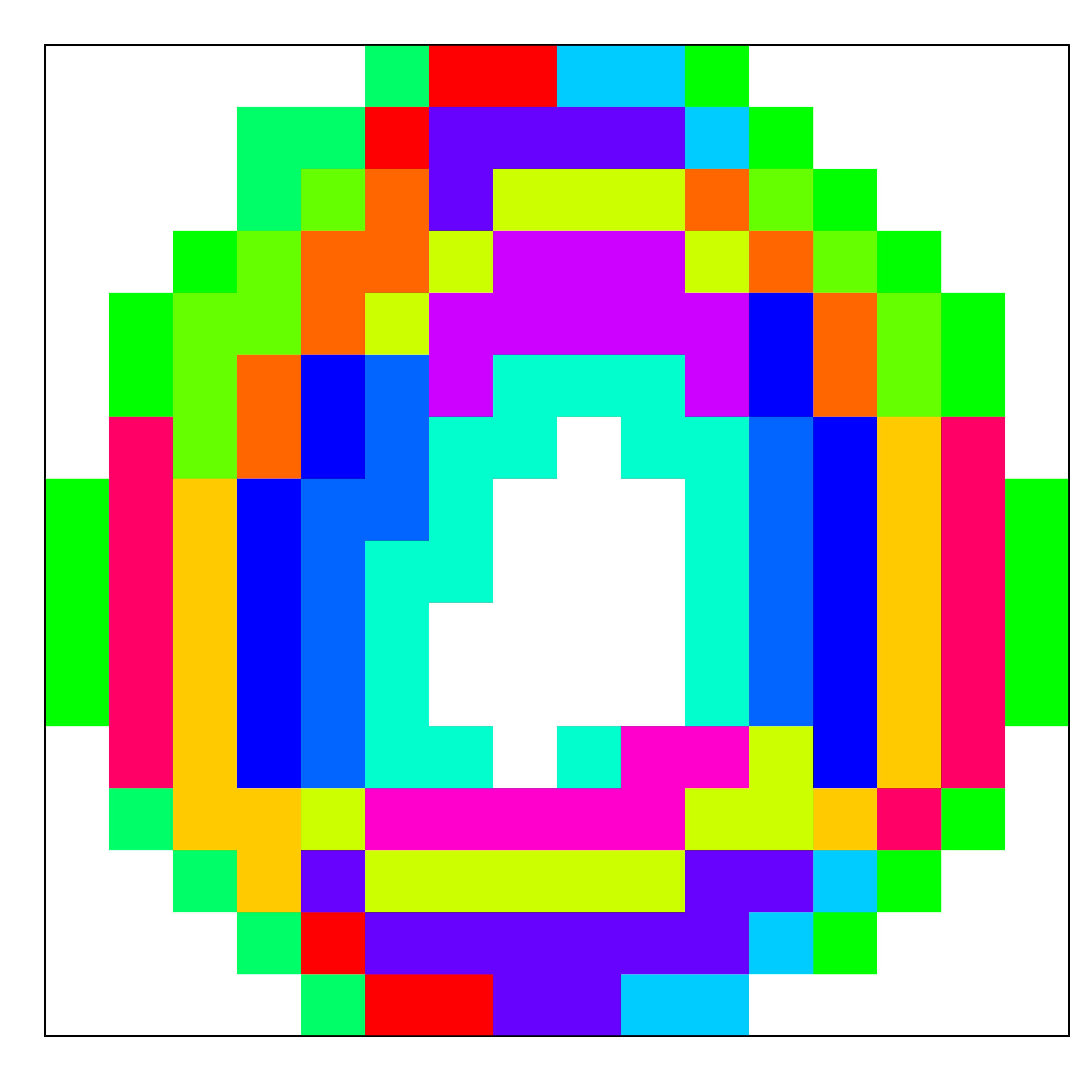}
  \caption{Results for FA (upper panels) and $k$-means (lower panels) when $m = 5$ (left panels), 10 (center panels), and 15 (right panels).}	 \label{fig:penalties}
\end{figure}

\begin{figure}[!t]
\centering
\includegraphics[width=4.5cm]{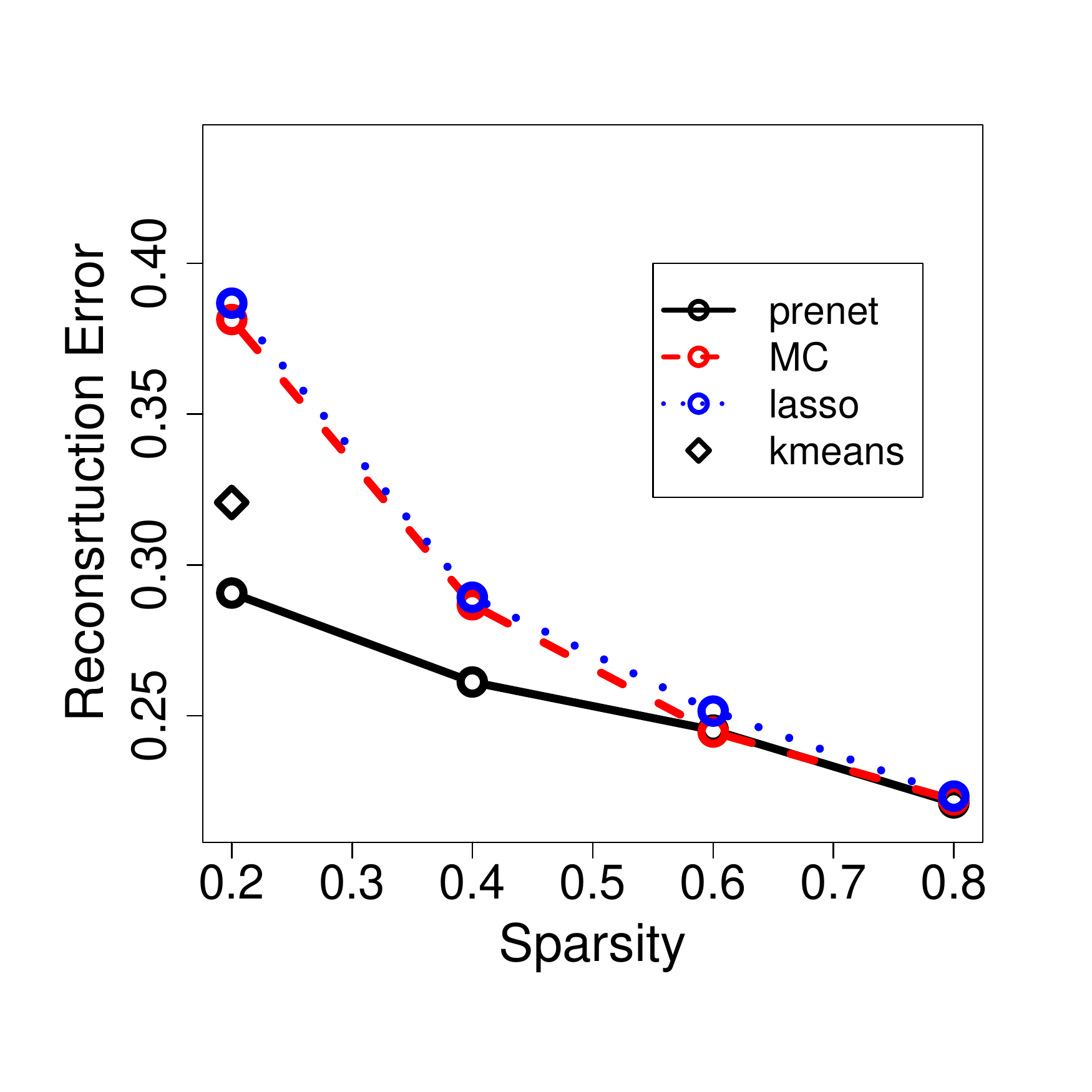}\hspace{2mm}
\includegraphics[width=4.5cm]{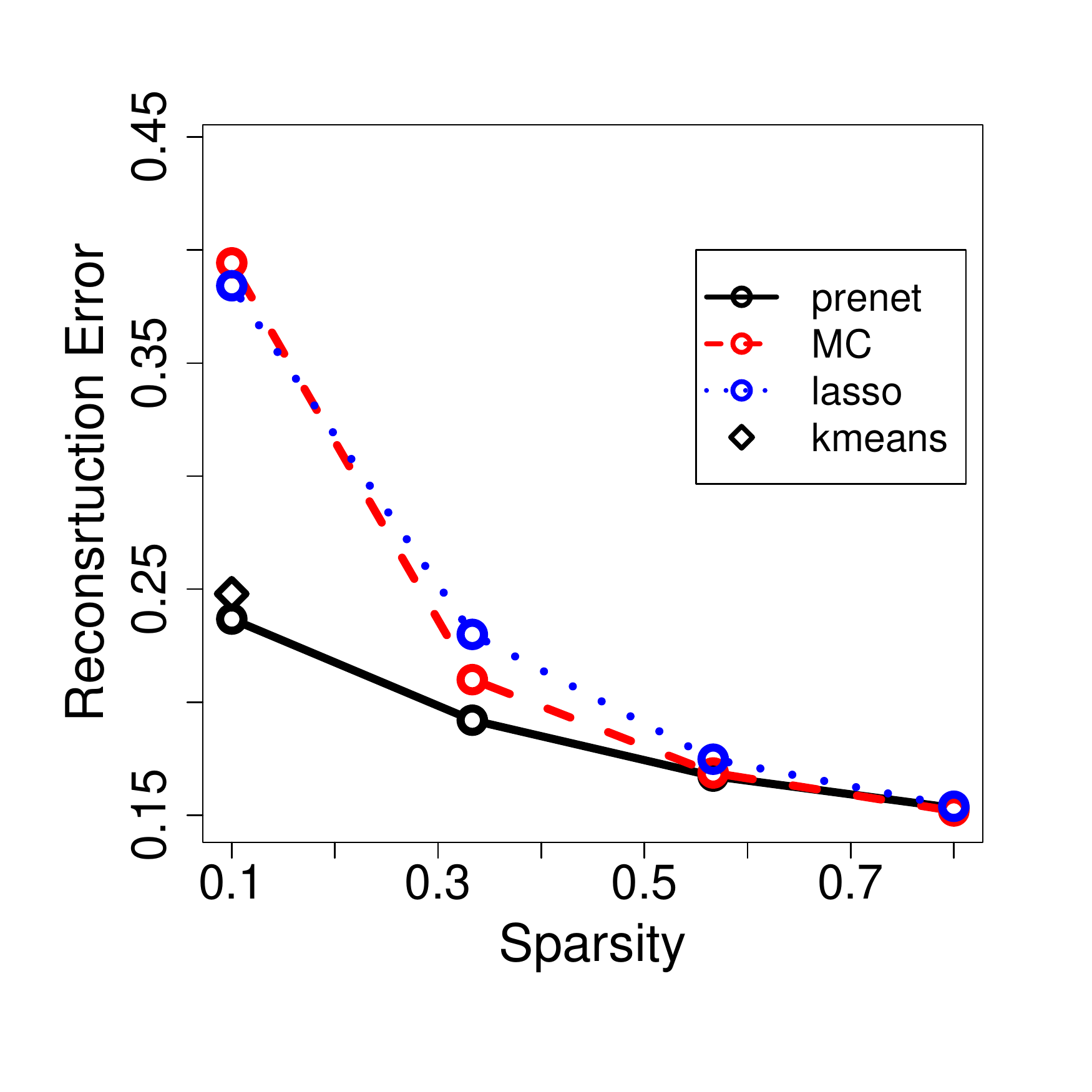}\hspace{2mm}
\includegraphics[width=4.5cm]{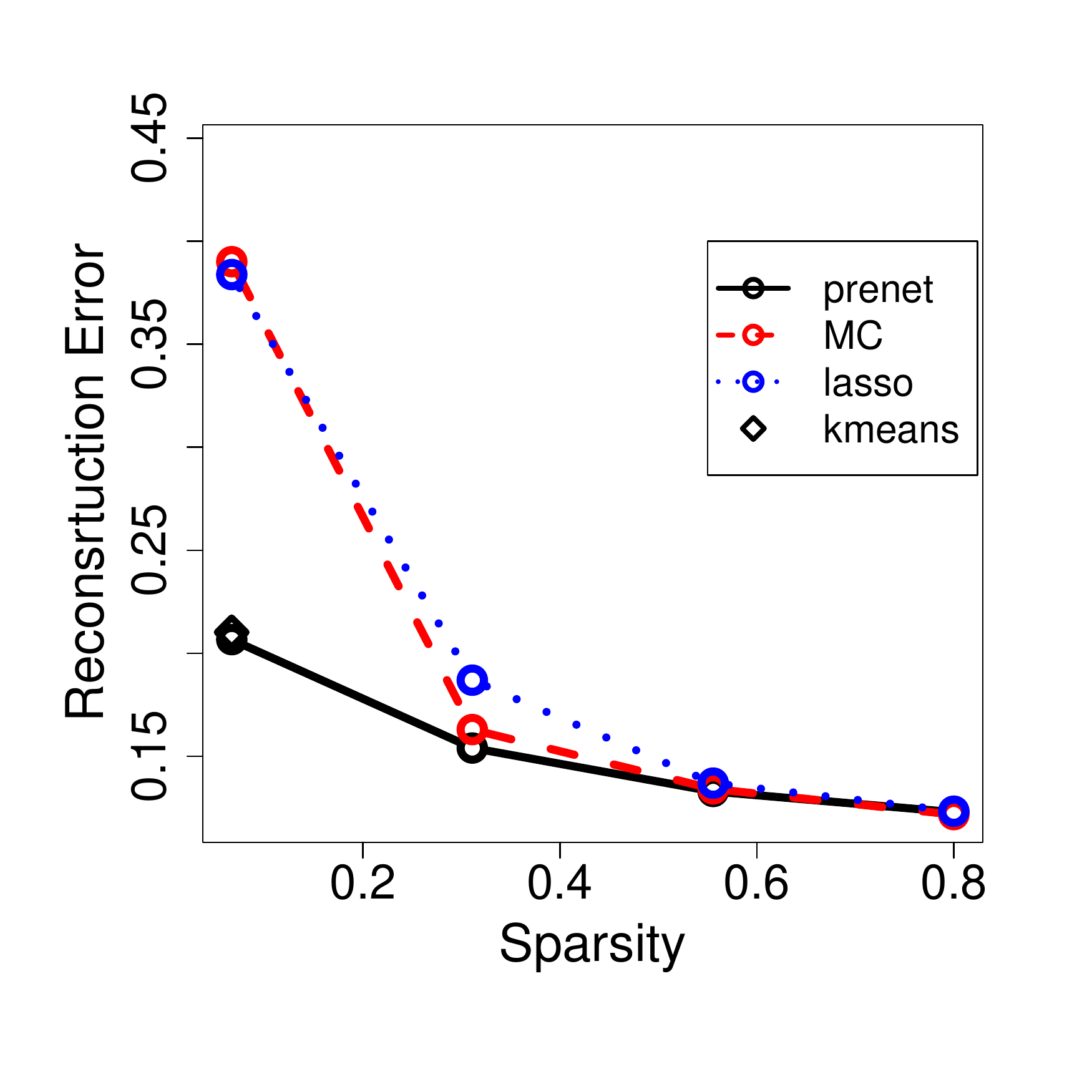}
  \caption{Reconstruction error when the number of factors (clusters) is $5$ (left panel), 10 (center panel), and 15 (right panel). The $x$ axis indicates the degrees of sparsity, and the $y$ axis indicates the reconstruction errors.}\label{fig:reconstruction error}
\end{figure}

We also compare the reconstruction error.  For $k$-means clustering, the data reconstruction of $\bm{x}_t$ is achieved using $\bm{\Lambda}(\bm{\Lambda}^T\bm{\Lambda})^{-1}\bm{\Lambda}^T\bm{x}_t$ $(t=1,\dots,n)$, where $\bm{\Lambda}$ is the estimated loading matrix.  In the prenet penalty, the data are reconstructed via the posterior mean:
\begin{eqnarray*}
\bm{\Lambda}E[\bm{F}_t | \bm{x}_t] = \bm{\Lambda}\bm{M}^{-1}\bm{\Lambda}^T\bm{\Psi}^{-1}\bm{x}_t \quad (t=1,\dots,n).
\end{eqnarray*}
We compress 359 test data with the above two methods and evaluate the performance by the reconstruction error.  We also compare the performance of above-mentioned two methods with that of the lasso and MC penalties.  The result is presented in Figure \ref{fig:reconstruction error}.

In the case of $m = 5$, the prenet penalty performs the best in terms of reconstruction error when the degree of sparsity is 0.2.  The second best method is the $k$-means, which implies the prenet results in a better cluster structure than the $k$-means in terms of reconstruction error.  The sparse estimations, such as the lasso and MC, perform very poorly.  We observe that  the lasso and MC result in a 3-factor model; the last two column vectors of the loading matrix result in $\bm{0}$.  For small degrees of sparsity, it is better to use the prenet penalty.  As the degrees of sparsity increase, the performance of the lasso and MC is competitive to that of the prenet.

When $m$ is large, the performance of the prenet with the sparsest model (i.e., perfect simple structure) is slightly better than that of the $k$-means but almost equivalent.  Interestingly, both lasso and MC perform poorly with small degrees of sparsity.  As the degrees of sparsity increase, the performance of the lasso and MC improve considerably and then become equivalent to that of the prenet.

\subsection{Resting state fMRI data}
In the third real data example, we investigate a cluster structure of brain regions of interest (ROIs) using a resting-state fMRI (rfMRI) data.  
We use a single-subject preprocessed resting-state fMRI data in Human Connectome Project (\url{https://www.humanconnectome.org/}).
The rfMRI data are acquired in a single run of 1200 time points (approximately 15 minutes).  
We view 268 brain regions proposed by \citet{shen2013groupwise} as ROIs, and aggregate the preprocessed voxel-wise rfMRI data into the 268 dimensional ROI-wise time series data by taking an average in each region.  

In this real data analysis, we conduct cluster analysis of the 268 ROIs.  Because the cluster analysis is an unsupervised learning, it is difficult to define a true cluster.  We consider target clusters as 8 clusters defined by \citet{finn2015functional}.   These 8 clusters are interpretable and determined by the group analysis of 126 subjects \citep{finn2015functional}.  
On the other hand, we use a single-subject resting-state fMRI data with 268 regions.   
We conduct a clustering by 
\begin{itemize}
\item Ward's method based on correlations among 268 ROIs,
\item perfect simple structure estimation via prenet penalization with 8 factors.
\end{itemize}
Note that we use $\xi_{ij}=1-|r_{ij}|$ as a dissimilarity between $i$th region and $j$th region on Ward's method,
where $r_{ij}$ is a correlation between time series of $i$th region and that of $j$th region.

Figure~\ref{fig:cluster} shows the clusters defined by \citet{finn2015functional}  and the results of both Ward's method and prenet penalization.
In each subfigure, the colored points are located at the center coordinates of the corresponding ROIs.  
Same color is corresponding to same cluster, so that colors of ROIs represent clusters. 
On the results of Ward's method and prenet penalization, the color combinations are chosen by matching the colors of clusters of \citet{finn2015functional}  as much as possible.
In order to compare these results more precisely, 
we use the adjusted Rand index (ARI), which is a measure of the similarity between two clustering results.  The larger the value of ARI, the higher the similarity between two clustering results is. 
The values of ARI between the two clustering results are given as follows:
\begin{itemize}
\item Ward's method and definition of \citet{finn2015functional}: $0.177$
\item prenet penalization and definition of \citet{finn2015functional}: $0.208$
\end{itemize}
Because the clusters defined by \citet{finn2015functional} are interpretable, the result shows that the prenet penalization may result in more interpretable clusters than the Ward's method.

\begin{figure}[h!]
 \begin{minipage}{0.32\textwidth}
  \centering
   \includegraphics[width=\textwidth]{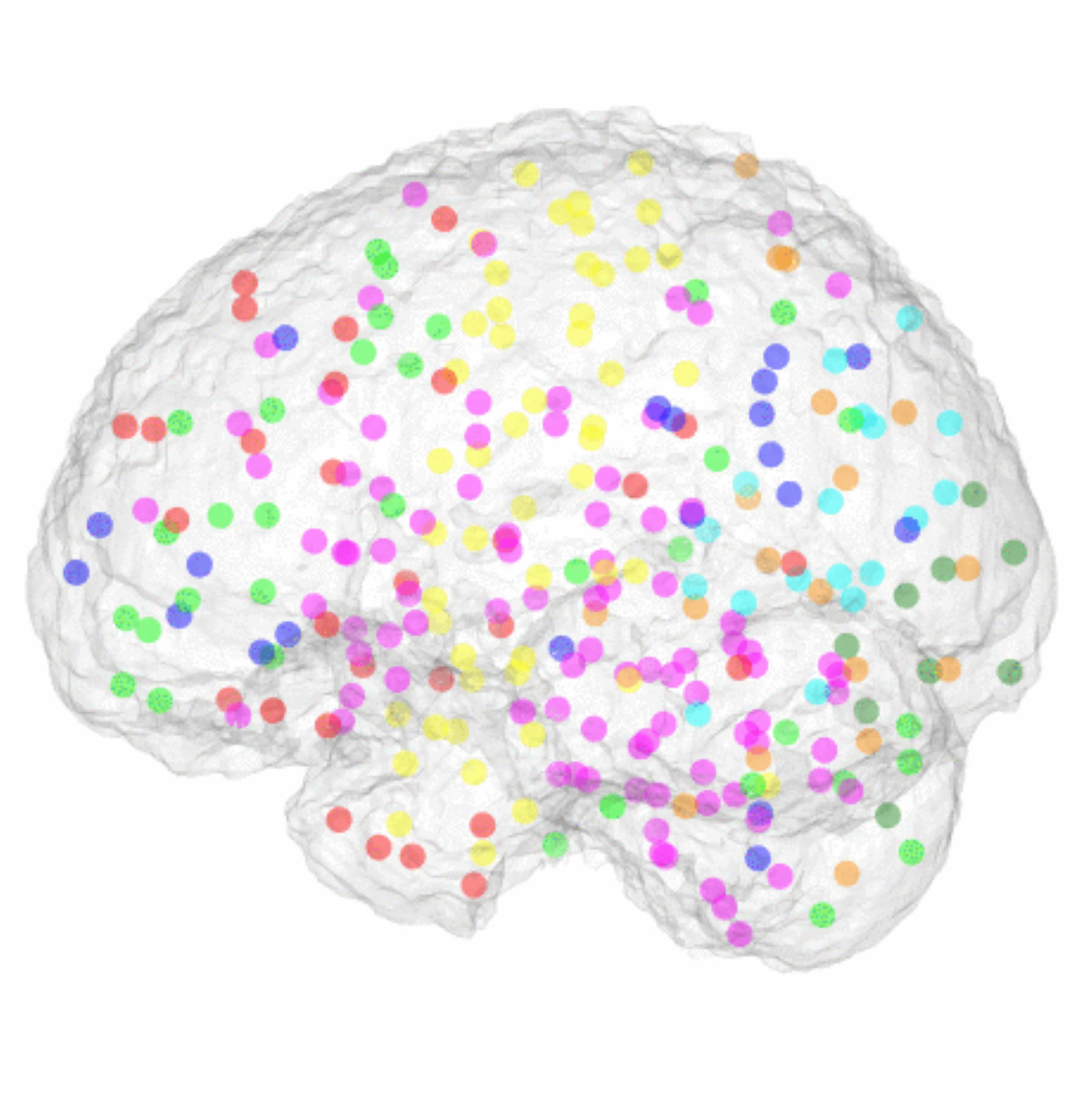}
   \subcaption{Finn et al. (2015)}
 \end{minipage}
 \begin{minipage}{0.32\textwidth}
  \centering
   \includegraphics[width=\textwidth]{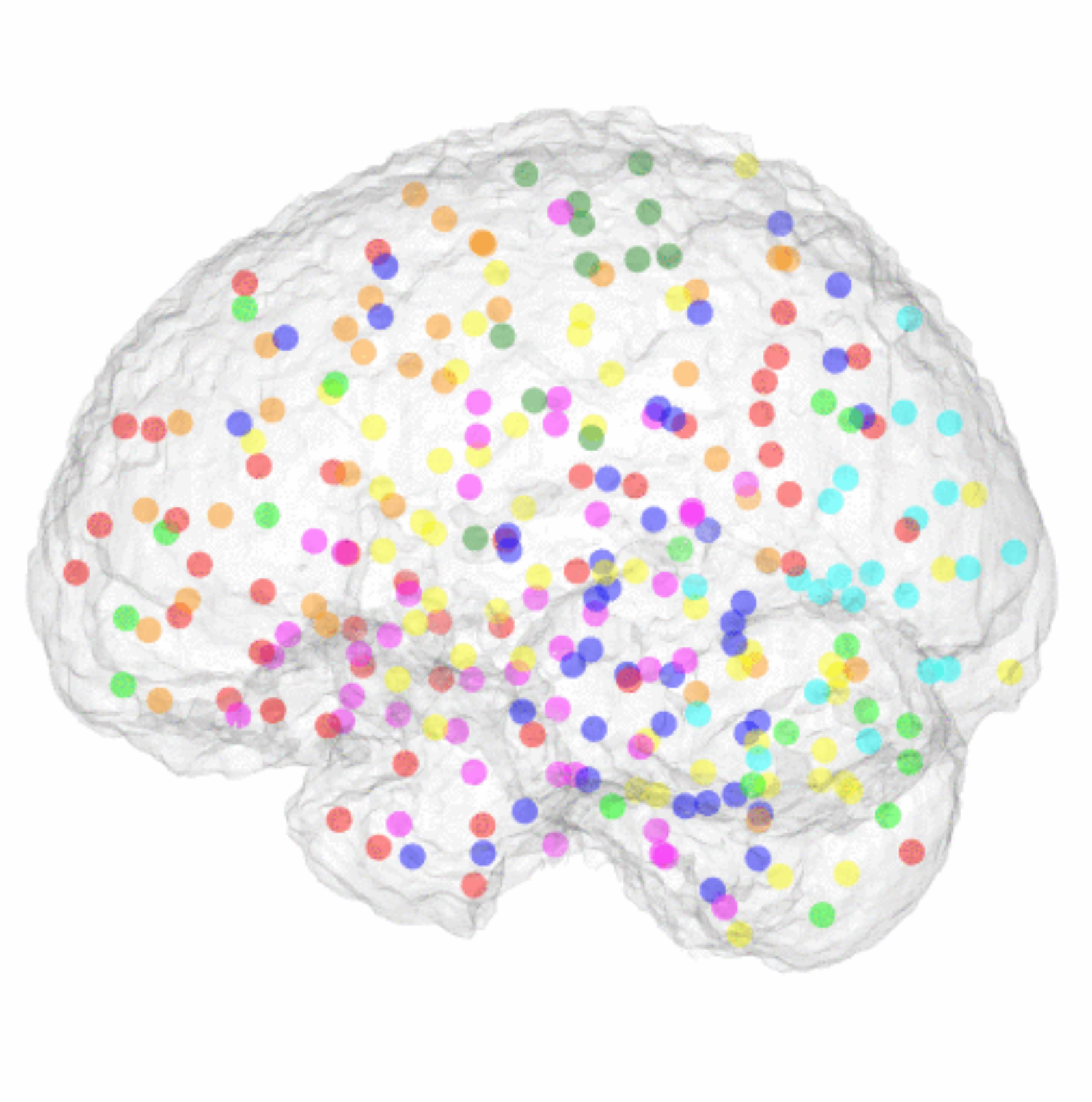}
   \subcaption{Ward's method}
 \end{minipage}
  \begin{minipage}{0.32\textwidth}
  \centering
   \includegraphics[width=\textwidth]{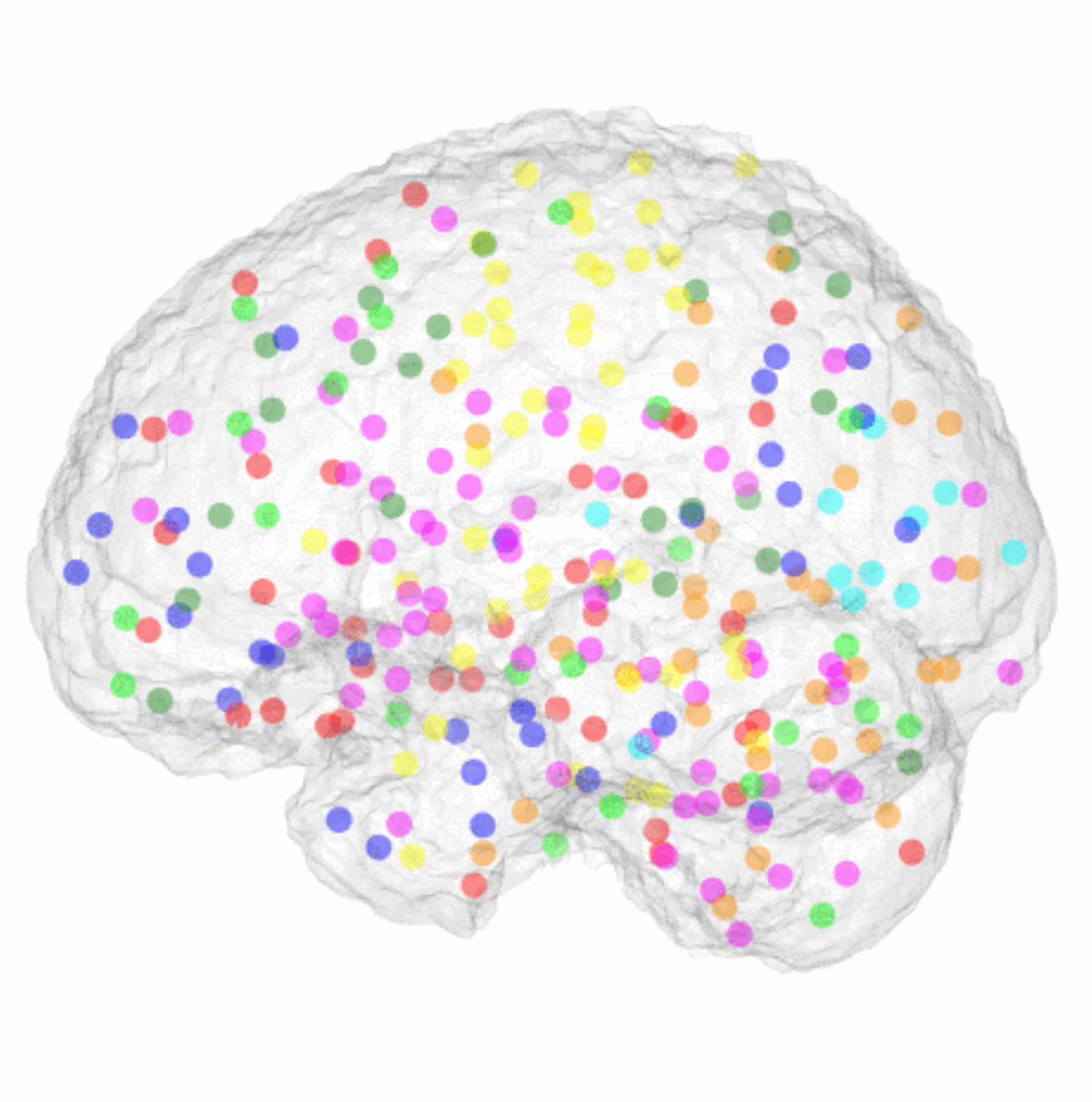}
   \subcaption{Prenet with 8 factors}
    \end{minipage}
   \caption{8 clusters of 268 ROIs.}
  \label{fig:cluster}
\end{figure}

\section{Concluding remarks}
We proposed a prenet penalty, which is based on the product of a pair of parameters in each row of the loading matrix.  The prenet penalty produced  the perfect simple structure for large values of $\rho$, which gave us a new variables clustering method using factor models.  In real data analysis, we showed that the prenet was able to capture a complex latent structure and outperformed the $k$-means in terms of reconstruction error.

The proposed penalty can be applied to any low rank matrix factorization, such as principal component analysis (PCA), non-negative matrix factorization, and so on.  In particular, the orthogonal nonnegative matrix factorization may be related to our method, because it corresponds to the perfect simple structure \citep{ding2005equivalence}.  The sparse PCA \citep{zou2006sparse} also assumes the orthogonality of the loading matrix, but some rows become zero vectors with a large amount of penalty.  It is interesting to apply the prenet penalty to other low rank matrix factorization methods, and compare the performance of the prenet with that of the existing estimation procedures.

The proposed method performed worse than sparse penalization, such as in the case of the MC penalty when the true loading matrix did not possess the perfect simple structure, as shown in Section \ref{sec:simulation}.  As described in \citet{yamamoto2013cluster}, the loading matrix does not always possess the perfect simple structure but it often has a well-clustered structure.  In such a case, a different penalty must be used.   In future research, it would be interesting to introduce a different penalty that captures more complex cluster structure than the perfect simple structure.

\appendix
\renewcommand{\theequation}{A\arabic{equation}}
\def\thesection{Appendix \Alph{section}}
\setcounter{equation}{0}

\section*{Acknowledgments}
The author would like to thank Dr. Michio Yamamoto for his guidance and suggestions.  This work was supported by a grant from Japan Society for the Promotion of Science KAKENHI 15K15949.

\section{Proofs}
\subsection{Proof of Proposition \ref{prop_kmeans}} \label{prop_kmeans:app}
Because of Proposition \ref{prop:pss}, with the prenet, $\hat{\lambda}_{ij}\hat{\lambda}_{ik} = 0$ as $\rho \rightarrow \infty$.  Thus, the prenet solution satisfies (\ref{q_ij_adj}) as $\rho \rightarrow \infty$.  We only need to show that the minimization problem of loss function $\ell_{\rm ML}(\bm{\Lambda},\bm{\Psi})$ is equivalent to that of $\| \bm{S} - \bm{\Lambda}\bm{\Lambda}^T\|^2$.  The inverse covariance matrix of the observed variables is expressed as
\begin{equation*}
\bm{\Sigma}^{-1} = \bm{\Psi}^{-1} - \bm{\Psi}^{-1}\bm{\Lambda}(\bm{\Lambda}^T\bm{\Psi}^{-1}\bm{\Lambda} + \bm{I})^{-1}\bm{\Lambda}^T\bm{\Psi}^{-1}.
\end{equation*}
Because $\bm{\Lambda}^T\bm{\Lambda} = \bm{I}_m$, we obtain
\begin{equation*}
\bm{\Sigma}^{-1} = \alpha^{-1}\bm{I} - \frac{\alpha^{-2}}{\alpha^{-1} + 1}\bm{\Lambda}\bm{\Lambda}^T.
\end{equation*}
The determinant of $\bm{\Sigma}$ can be calculated as
\begin{equation*}
|\bm{\Sigma}| = \alpha^{p-m}(1+\alpha)^m.
\end{equation*}
Then, the discrepancy function in (\ref{taisuuyuudo}) is expressed as
\begin{equation*}
	\frac{1}{2}\left\{ {\rm tr}(\alpha^{-1}\bm{S}) - \frac{\alpha^{-2}}{\alpha^{-1} + 1}{\rm tr}\left( \bm{\Lambda}^T\bm{S}\bm{\Lambda} \right) + p\log \alpha + m \log \left( 1+\frac{1}{\alpha} \right)  - \log|\bm{S}| -p  \right\}.
\end{equation*}
Because $\alpha$ is given and $\| \bm{S} - \bm{\Lambda}\bm{\Lambda}^T\|^2 = -2{\rm tr} \left( \bm{\Lambda}^T\bm{S}\bm{\Lambda} \right) + {\rm const}.$, we can derive (\ref{problem_kmeansgen}).

\subsection{Proof of Proposition \ref{prop_qmin}}\label{qmin:app}
Recall that $\hat{\bm{\theta}}=(\hat{\bm{\Lambda}},\hat{\bm{\Psi}})$ is an unpenalized estimator that satisfies $\displaystyle {\ell}(\hat{\bm{\theta}})=\min_{\bm{\theta} \in \Theta} {\ell}(\bm{\theta})$ and $\hat{\bm{\theta}}_q$ is a quartimin solution obtained by the following problem:
$$
\min_{\bm{\theta} \in \Theta}P_{\rm qmin}(\bm{\Lambda}), \mbox{ subject to} \quad   \ell(\bm{\theta}) = \ell(\hat{\bm{\theta}}).
$$
First, we show that
\begin{equation}
\lim_{n\rightarrow \infty} d(\hat{\bm{\theta}}_q,\Theta_q^\ast)=0\;\;\text{a.s.}\label{eq:qmin}
\end{equation}
From the assumptions, as the same manner of Chapter 6 in \citet{pfanzagl1994parametric}, 
we can obtain the following strong consistency:
\begin{equation}
\lim_{n\rightarrow \infty} d(\hat{\bm{\theta}},\Theta_\ast)=0\;\text{ and }\;\lim_{n\rightarrow \infty} d(\hat{\bm{\theta}}_{\rho_n},\Theta_\ast)=0\quad\text{a.s.} \label{strong consistency}	
\end{equation}
where 
$\Theta_\ast:=\{\bm{\theta}\in \Theta\mid \ell_\ast(\bm{\theta})=\min_{\bm{\theta}\in \Theta}\ell_\ast(\bm{\theta})\}$.  
$\lim_{n\rightarrow \infty} d(\hat{\bm{\theta}},\Theta_\ast)=0$ implies for all $\epsilon>0$, by taking $n$ large enough, we have 
$$
\|\hat{\bm{\Lambda}}-\bm{\Lambda}_\ast\|<\epsilon
\quad\text{a.s.}
$$
 for some $(\bm{\Lambda}_*, \bm{\Psi}_*) \in \Theta_*$.  From the uniform continuity of $P_\mathrm{qmin}$ on $\Theta$ and the fact that $\|\hat{\bm{\Lambda}}\bm{T}-\bm{\Lambda}_\ast \bm{T}\|=\|\hat{\bm{\Lambda}}-\bm{\Lambda}_\ast\|$ for any $\bm{T}\in \mathcal{O}(m)$, 
we have
\begin{equation}
\sup_{\bm{T}\in \mathcal{O}(m)}
|P_\mathrm{qmin}(\hat{\bm{\Lambda}}\bm{T})-P_\mathrm{qmin}(\bm{\Lambda}_\ast\bm{T})|<\epsilon
\quad\text{a.s.}
\label{qmin:sup}	
\end{equation}
Write $\hat{\bm{T}}:=\mathop{\arg\min}_{\bm{T}\in \mathcal{O}(m)}P_\mathrm{qmin}(\hat{\bm{\Lambda}}\bm{T})$ and $\bm{T}_\ast:=\mathop{\arg\min}_{\bm{T}\in \mathcal{O}(m)}P_\mathrm{qmin}(\bm{\Lambda}_\ast\bm{T})$. We have 
$$
P_\mathrm{qmin}(\hat{\bm{\Lambda}}\hat{\bm{T}})-P_\mathrm{qmin}(\bm{\Lambda}_\ast\hat{\bm{T}})
\le
P_\mathrm{qmin}(\hat{\bm{\Lambda}}\hat{\bm{T}})-P_\mathrm{qmin}(\bm{\Lambda}_\ast\bm{T}_\ast)
\le
P_\mathrm{qmin}(\hat{\bm{\Lambda}}\bm{T}_\ast)-P_\mathrm{qmin}(\bm{\Lambda}_\ast\bm{T}_\ast).
$$
From this, it follows that
$$
|P_\mathrm{qmin}(\hat{\bm{\Lambda}}\hat{\bm{T}})-P_\mathrm{qmin}(\bm{\Lambda}_\ast\bm{T}_\ast)|
\le 
\sup_{\bm{T}\in \mathcal{O}(m)}
|P_\mathrm{qmin}(\hat{\bm{\Lambda}}\bm{T})-P_\mathrm{qmin}(\bm{\Lambda}_\ast\bm{T})|.
$$
Thus, using (\ref{qmin:sup}), we obtain (\ref{eq:qmin}).

Next, as the similar manner of Proposition 15.1 in \citet{foucart2013mathematical}, 
we prove $\lim_{n\rightarrow \infty} d(\hat{\bm{\theta}}_{\rho_n},\Theta_q^\ast)=0\;\;\text{a.s.}$
By the definition of $\hat{\bm{\theta}}_{\rho_n}$, for any $\rho_n>0$ we have 
\begin{equation}
\ell(\hat{\bm{\theta}}_{{\rho_n}}) + \rho_n P_{\rm qmin}(\hat{\bm{\Lambda}}_{{\rho_n}}) \le \ell(\hat{\bm{\theta}}_q) +\rho_n P_{\rm qmin}(\hat{\bm{\Lambda}}_q) \label{ineq}
\end{equation}
and 
\begin{equation}
	\ell(\hat{\bm{\theta}}_{{\rho_n}}) \ge \ell(\hat{\bm{\theta}}_q). \label{ineq2}
\end{equation}
Combining (\ref{eq:qmin},\ref{ineq},\ref{ineq2}), we obtain 
\begin{equation}
    	 P_{\rm qmin}(\hat{\bm{\Lambda}}_{{\rho_n}})\le P_{\rm qmin}(\hat{\bm{\Lambda}}_q)\rightarrow 
    	 P_\mathrm{qmin}(\bm{\Lambda}^{\ast}_q)\quad\text{a.s.} \label{eq:final}
\end{equation}
for some $(\bm{\Lambda}^{\ast}_q,\bm{\Psi}^{\ast}_q) \in \Theta^*_q $.  Therefore, we have
$$
\mathop{\lim}_{n \rightarrow \infty} P_{\rm qmin}(\hat{\bm{\Lambda}}_{\rho_n})\le 
P_\mathrm{qmin}(\bm{\Lambda}^{\ast}_q)\quad\text{a.s.}
$$
As shown in (\ref{strong consistency}), $\lim_{n\rightarrow \infty} d(\hat{\bm{\theta}}_{\rho_n},\Theta_\ast)=0\;\;\text{a.s.}$, and $\bm{\Lambda}^{\ast}_q$ is a minimizer of $P_{\rm qmin}(\cdot)$ over $\Theta_*$, so that the proof is complete.

\section{Construction of the varimax penalty}
\renewcommand{\theequation}{B\arabic{equation}}
\setcounter{equation}{0}
The varimax criterion \citep{kaiser1958varimax} is expressed by
\begin{eqnarray*}
	Q(\bm{\Lambda}) &=& \sum_{k = 1}^m  \sum_{i = 1}^p \left\{ \lambda_{ik}^2 - \frac{1}{p} \left( \sum_{i = 1}^p \lambda_{ik}^2 \right) \right\}^2 = \sum_{k=1}^m \left\{ \sum_{i = 1}^p \lambda_{ik}^4 - \frac{1}{p} \left( \sum_{i = 1}^p \lambda_{ik}^2 \right)^2 \right\}. \label{varimax}
\end{eqnarray*}
However, we cannot directly apply the varimax rotation criterion $Q(\bm{\Lambda})$ as the penalty function $P(\bm{\Lambda})$, because the varimax criterion must be {\it maximized} under some constraint.  In other words, if the varimax criterion is used as a penalty of the penalized factor analysis, it must be
\begin{equation}
	\ell_{\rho}(\bm{\Lambda},\bm{\Psi}) = \ell(\bm{\Lambda},\bm{\Psi}) - \rho Q(\bm{\Lambda}). \label{varimax direct}
\end{equation}
It is easily shown that $Q(a\bm{\Lambda}) > Q(\bm{\Lambda})$ for any $a > 1$.  Thus, (\ref{varimax direct}) implies the estimate of factor loadings increase as $\rho$ increases.  Estimating coefficients that are too large are opposed to the basic concept of the penalization procedure; the penalization procedure usually shrinks some coefficients toward zero to produce stable estimates.

In order to overcome this problem, we consider the equivalent minimization problem of the varimax criterion.
\begin{eqnarray*}
	\sum_{k = 1}^m  \sum_{i=1}^p \lambda_{ik}^4 &=& \sum_{k = 1}^m \sum_{l = 1}^p \sum_{i = 1}^p \lambda_{ik}^2 \lambda_{il}^2 - \sum_{k = 1}^m \sum_{l\ne k}^p \sum_{i = 1}^p \lambda_{ik}^2 \lambda_{il}^2\\
	&=& \sum_{i = 1}^p \left( \sum_{k = 1}^m \lambda_{ik}^2 \right) \left( \sum_{l = 1}^p \lambda_{il}^2 \right) - \sum_{k = 1}^m \sum_{l\ne k}^p \sum_{i = 1}^p \lambda_{ik}^2 \lambda_{il}^2
\end{eqnarray*}
Here, the value of $\sum_{k = 1}^m \lambda_{ik}^2$ is invariant with respect to the orthogonal rotation. Therefore, maximization of (\ref{varimax}) over all loading matrices of the maximum likelihood estimate is equivalent to the minimization of the following function:
\begin{equation}
P(\bm{\Lambda}) = \sum_{k = 1}^m  \sum_{l \ne k} \sum_{i = 1}^p \lambda_{ik}^2\lambda_{il}^2 + \frac{1}{p}  \sum_{k = 1}^m  \left( \sum_{i =1 }^p \lambda_{ik}^2 \right)^2. \label{varimax_adj}
\end{equation}
We may use (\ref{varimax_adj}) as a penalty function of the penalized factor analysis.  

\section{Update equation via the coordinate descent algorithm}
\renewcommand{\theequation}{C\arabic{equation}}
\setcounter{equation}{0}
Let $\tilde{\bm{\lambda}}_{i}^{(j)}$ be a ($m-1$)-dimensional vector $( \tilde{\lambda}_{i1},\tilde{\lambda}_{i2},\dots,\tilde{\lambda}_{i(j-1)},\tilde{\lambda}_{i(j+1)},\dots,\tilde{\lambda}_{im})^T$.  The parameter $\lambda_{ij}$ can be updated by maximizing (\ref{ECL}) with the other parameters $\tilde{\bm{\lambda}}_{i}^{(j)}$ and with $\bm{\Psi}$ being fixed, that is, we solve the following problem:
\begin{eqnarray}
\tilde{\lambda}_{ij} &=& {\rm arg} \min_{\lambda_{ij}}  \frac{1}{2\psi_i} \left\{a_{jj}\lambda_{ij}^2  - 2\left(b_{ij} - \sum_{k \ne j} a_{kj} \tilde{\lambda}_{ik} \right)\lambda_{ij} \right\}  \cr
&& + \rho   \left[ \left\{ \frac{1}{2} (1-\gamma) \sum_{k \neq j} \tilde{\lambda}_{ik}^2\right\}\lambda_{ij}^2 + \left( \gamma  \sum_{k \neq j} |\tilde{\lambda}_{ik}| \right)|\lambda_{ij}|  \right]  \cr
 &=& {\rm arg} \min_{\lambda_{ij}}  \frac{1}{2\psi_i} \left\{(a_{jj}+\beta)\lambda_{ij}^2  - 2\left(b_{ij} - \sum_{k \ne j} a_{kj} \tilde{\lambda}_{ik} \right)\lambda_{ij} \right\}  + \rho\xi|\lambda_{ij}| \cr
&=& {\rm arg} \min_{\lambda_{ij}}\frac{1}{2}   \left( \lambda_{ij}  - \frac{b_{ij} - \sum_{k \ne j} a_{kj} \tilde{\lambda}_{ik} }{a_{jj} + \beta} \right)^2  + \frac{\psi_i\rho\xi}{a_{jj}+\beta}  |\lambda_{ij}|. \label{lambdaupdate}
\end{eqnarray}
where
\begin{eqnarray*}
\beta &=& \rho\psi_i(1-\gamma) \sum_{k \neq j} \tilde{\lambda}_{ik}^2, \cr
\xi &=& \gamma  \sum_{k \neq j} |\tilde{\lambda}_{ik}|.
\end{eqnarray*}
This is equivalent to minimizing the following penalized squared error loss function
\begin{equation*}
S(\tilde{\theta}) = {\rm arg} \min_{\theta} \left\{ \frac{1}{2}(\theta - \tilde{\theta})^2 + \rho^* |\theta| \right\}. \label{lamdba_update_CD}
\end{equation*}
The solution $S(\tilde{\theta})$ can be expressed in a closed form using the following soft thresholding function.
\begin{equation*}
S(\tilde{\theta})= {\rm sgn}(\tilde{\theta}) (|\tilde{\theta} |- \rho^*)_+, \label{uelasso}
\end{equation*}
where $A_+ = \max(A,0)$.

\bibliographystyle{rss}
\bibliography{paper-ref}

\end{document}